%% file: main.tex
\documentclass[preprint]{elsarticle}

\usepackage{lineno,hyperref}

\usepackage{amsmath}
\usepackage{amssymb}
\usepackage{xspace}
\usepackage{proof-dashed}
\usepackage{tikz}
\usepackage[arrow,matrix,curve]{xy}

\input{macros}

\journal{Journal of \LaTeX\ Templates}

\begin{document}

\begin{frontmatter}

  \title{Undecidability of Asynchronous Session Subtyping}

  \author{Mario Bravetti
  } \address{University of Bologna, Department of Computer Science and Engineering / FOCUS INRIA\\ Mura Anteo Zamboni 7, 40126 Bologna, Italy }

  \author{Marco Carbone
  } \address{Department of Computer Science, IT University of
    Copenhagen\\ Rued Langgaards Vej 7, 2300 Copenhagen, Denmark}

\author{Gianluigi Zavattaro
  } \address{University of Bologna, Department of Computer Science and Engineering / FOCUS INRIA\\ Mura Anteo Zamboni 7, 40126 Bologna, Italy }




\begin{abstract}
  Session types are used to describe communication protocols in
  distributed systems and, as usual in type theories, session subtyping 
  characterizes 
  substitutability of the communicating processes. 
  We investigate the (un)decidability of subtyping for 
  session types in asynchronously communicating systems.
  We first devise a core undecidable subtyping relation that is 
  obtained by imposing limitations on the 
  structure of types. 
  Then, as a consequence of this initial undecidability result,
  we show that (differently from what
  stated or conjectured in the literature) the three
  notions of asynchronous subtyping defined so far for session types 
  are all undecidable. Namely, we consider
  the asynchronous session subtyping by Mostrous and
  Yoshida~\cite{IandCmostrous} for binary sessions, the relation by
  Chen et al.~\cite{MariangiolaPreciness} for binary sessions under
  the assumption that every message emitted is eventually consumed,
  and the one by Mostrous et al.~\cite{ESOP09} for multiparty session
  types.
%
%
%
  Finally,  
  by showing that two fragments of the core
  subtyping relation are decidable, we
  evince that
  further restrictions on the 
  structure of types 
  make
  our core subtyping relation decidable.
%
%
%
%
%

\end{abstract}


\end{frontmatter}


\section{Introduction}
\label{sec:introduction}
\input introduction

\section{Asynchronous Subtyping}
\label{sec:subtyping}
\input{subtyping}
\section{Core Undecidability Result}
\label{sec:undecidable}
\input{undecidable}

\section{Impact of the Undecidability Result}
\label{subsec:impact}
\input impact

\section{Decidable Fragments of Aynchronous \Selection Relation}
\label{sec:decidable}
\input{decidable}

\section{Conclusion and Related Work}
\label{subsec:conclusions}
\input conclusions

\section*{References}
\bibliography{biblio}

\end{document}

%% file: macros.tex
\newcommand{\selection}{single-choice\xspace}
\newcommand{\Selection}{Single-Choice\xspace}

\newcommand{\calt}{{\cal T}}

\newcommand{\infr}[2]
        {\renewcommand{\arraystretch}{1.5}
        \begin{array}{c}
        #1\\
        \hline
        #2
        \end{array}}

\newcommand{\auxarrow}
        {\mathop{\longrightarrow}}

\newcommand{\arrow}[1]
        {\, \auxarrow\limits^{#1} \,}

\newtheorem{definition}{Definition}[section]
\newtheorem{corollary}{Corollary}[section]
\newtheorem{theorem}{Theorem}[section]
\newtheorem{lemma}{Lemma}[section]
\newcommand{\proof}{{\em Proof.}\ }

\newcommand{\Tout}[1]{!\langle {#1}\rangle}
\newcommand{\Tin}[1]{?({#1})}

\newcommand{\ToutK}[1]{k!\langle {#1}\rangle}
\newcommand{\TinK}[1]{k?({#1})}
\newcommand{\Tbranchindex}[4]{\&\{{#1}_{#3}:{#2}_{#3}\}_{#3\in #4}}
\newcommand{\Tbranch}[2]{\Tbranchindex{#1}{#2}{i}{I}}
\newcommand{\Tbranchsingledec}[3]{\&^{#3}\{{#1}:{#2}\}}
\newcommand{\Tbranchsingle}[2]{\&\{{#1}:{#2}\}}

\newcommand{\Tbranchset}[3]{\&\{{#1}\!:\!{#2}\}_{#1\in #3}}

\newcommand{\TbranchK}[2]{k\Tbranchindex{#1}{#2}{i}{I}}

\newcommand{\TbranchindexM}[5]{\&\{{#1}_{#3}(#5_{#3}):{#2}_{#3}\}_{#3\in #4}}
\newcommand{\TbranchM}[3]{\TbranchindexM{#1}{#2}{i}{I}{#3}}

\newcommand{\Tbranchsimple}[2]{\&\{{#1}:{#2}\}}

\newcommand{\Tselectindex}[4]{\oplus\{{#1}_{#3}:{#2}_{#3}\}_{#3\in #4}}
\newcommand{\Tselect}[2]{\Tselectindex{#1}{#2}{i}{I}}
\newcommand{\Tselectsingle}[2]{\oplus\{{#1}:{#2}\}}
\newcommand{\Tselectset}[3]{\oplus\{{#1}:{#2}\}_{#1\in #3}}

\newcommand{\TselectK}[2]{k\Tselectindex{#1}{#2}{i}{I}}

\newcommand{\TselectindexM}[5]{\oplus\{{#1}_{#3}\langle #5_{#3}\rangle:{#2}_{#3}\}_{#3\in #4}}
\newcommand{\TselectM}[3]{\TselectindexM{#1}{#2}{i}{I}{#3}}

\newcommand{\Tselectsimple}[2]{\oplus\{{#1}:{#2}\}}

\newcommand{\Trec}[1]{\mu \mathbf{#1}}
\newcommand{\Tvar}[1]{\mathbf{#1}}
\newcommand{\Tend}{\mathbf{end}}
\newcommand{\context}[3]{\mathcal{#1}[{#2}]^{#3}}

\newcommand{\rderiv}{\rightarrow}

\newcommand{\notrderiv}{\rightarrow_{\mathsf{err}}}

\newcommand{\unfold}[2]{\mathsf{unfold}^{#1}(#2)}
\newcommand{\depth}{\mathsf{depth}}
\newcommand{\noinf}{T^\mathsf{noinf}}

\newcommand{\selsubtype}{{<\!\!<}}
\newcommand{\subtype}{{\leq}}
\newcommand{\subtypea}{\;{\leq_{\mathsf{a}}}\;}
\newcommand{\subtypet}{\;{\leq_{\mathsf{t}}}\;}
\newcommand{\semT}[2]{[\!\![{#1}]\!\!]^{#2}}
\newcommand{\semS}[1]{[\!\![{#1}]\!\!]}

\newcommand{\decore}[1]{\mathsf{ann}(#1)}
\newcommand{\undecore}[1]{\mathsf{unann}(#1)}

\newcommand{\recclose}[1]{\mathsf{topUnfold}(#1)}


%% file: introduction.tex
Session types~\cite{HVK98,HYC16} are types for describing the
behaviour of communicating systems, and can be used as specifications 
of distributed protocols to be checked against implementations.
Such check, done by means of a typing system, guarantees that
communications at any endpoint of the implemented system are always
matched by the corresponding intended partner. As a consequence, it is
ensured that communication errors, e.g., deadlock, will never
occur. This approach provides a compositional way of checking the
correctness of distributed systems.

As an example, consider a simple on-line shop: clients can
buy a list of items by following the protocol expressed
by the session type
$$S_{client} = \Trec t.\oplus\{\text{add\_to\_cart}: \Tvar t\ ,\ \text{pay}: \Tend\}$$
indicating a recursive behaviour, according to which the client decides whether to add
an item and keep interacting with the store, or to pay
and conclude the session. 
For the sake of simplicity we consider session types where
(the type of) communicated data is abstracted away.

We call \emph{output selection}\footnote{In session type terminology \cite{HVK98,HYC16}, 
this construct is usually simply called \emph{selection}; we call it output selection
because we consider a simplified syntax for session types in which there
is no specific separate construct for sending one output. 
Anyway, such an output type 
could be seen as an output selection with only one choice.}
the construct $\oplus\{l_1:T_1, \dots, l_n:T_n\}$. It is used to
denote a point of choice in the communication protocol:
each choice has a label $l_i$ and a continuation $T_i$. 
In communication protocols,
when there is a point of choice, there is usually a peer that
internally takes the decision and the other involved peers receive
communication of the selected branch.
Output selection is used to describe the behaviour of the peer that takes the decision:
indeed, in our example it is the client that decides when to stop
adding items to the cart and then move to the payment.

The symmetric behaviour of the
shopping service is represented by the complementary session type
$$S_{service} = \Trec t.\& \{\text{add\_to\_cart}: \Tvar t\ ,\ \text{pay}: \Tend\}.$$
We call \emph{input branching}\footnote{In session type terminology this
construct is simply called \emph{branching}. We call it input branching
for symmetric reasons w.r.t. those discusses in the previous footnote.} the construct $\&\{l_1:T_1, \dots, l_n:T_n\}$. 
It is used to describe the behaviour
of a peer that receives communication of the selection
done by some other peers. In the example, indeed, the service receives
from the client the decision about the selection.



When composing systems whose interaction protocols have been specified
with session types, it is 
significant to consider variants
of their specifications that still preserve safety properties. 
In the above example of the on-line shop, the client can be
safely replaced by another one with session type
$$T_{client}=\oplus\{\text{add\_to\_cart}: \oplus \{\text{pay}: \Tend\}\}$$
indicating that only one item is added to the shopping cart
before paying.
But also the shopping service could be safely replaced by another one
offering also the $\text{remove\_from\_cart}$ functionality:
$$T_{service}=\Trec t.\& \{\text{add\_to\_cart}: \Tvar t\ ,\ \text{remove\_from\_cart}: \Tvar t\ ,\ \text{pay}: \Tend\}.$$
Formally, subtyping relations have been defined for session types
to precisely capture this safe replacement notion.

Gay and Hole~\cite{GH05} are the first ones who studied subtyping for
session types in a context where protocols involve only two peers (i.e.\ are binary)
and communication is synchronous. Later, Mostrous et al.~\cite{ESOP09}
extended this notion to multiparty session types with {\em
  asynchronous} communication. Both articles propose an algorithm for
checking subtying, but
the one proposed by Mostrous et
al.~\cite{ESOP09}, differently from what stated therein, 
is not always terminating in the sense that there are cases
in which it diverges and never gives an answer.
An example of divergent execution is 
discussed in the \emph{Remark} paragraph of 
\S\ref{sec:multiparty}.

%
Later work 
by Mostrous and Yoshida~\cite{IandCmostrous},
Mostrous~\cite{dimitristhesis} and Chen et
al.~\cite{MariangiolaPreciness} addresses subtyping in variants of an
asynchronous setting for binary sessions. In particular~Chen et
al.~\cite{MariangiolaPreciness} focus on binary sessions in which
messages sent by a partner are guaranteed to be eventually received.  
Such articles
conjecture that an algorithm for checking asynchronous session
subtyping exists, although, in his PhD thesis,
Mostrous~\cite{dimitristhesis} expresses a few doubts about the
decidability of asynchronous subtyping (pp.~178-180), because of the
need for infinite simulations.

In this work, 
we prove that the subtyping relations defined 
by Mostrous and Yoshida~\cite{IandCmostrous},
Chen et al.~\cite{MariangiolaPreciness}, 
and Mostrous et al.~\cite{ESOP09} are undecidable.
We proceed by identifying a core asynchronous subtyping
relation and show it is undecidable: all other undecidability
results are obtained by reduction from this initial relation.

The core relation, denoted by $\selsubtype$, is named asynchronous
\selection relation. Such a relation is obtained by first defining
(following the approach by Mostrous and Yoshida~\cite{IandCmostrous})
a standard asynchronous subtyping 
$\subtype$ and then reduce it by imposing additional constraints: $T$
and $S$ are in \selection relation, written $T \,\selsubtype \,S$, if $T
\, \subtype \, S$, all output selections in $T$ have a single choice 
(output selections are covariant~\footnote{Covariant means that the
  bigger type has more choices than the smaller type.}, thus $S$ is
allowed have output selections with multiple choices), all input branchings in $S$ have 
a single choice (input branchings are
contravariant~\footnote{Contravariant means that the bigger type has
  less choices than the smaller type.}, thus $T$ is allowed to have
input branchings with multiple choices), and additionally both $T$
and $S$ do not have consecutive infinite output selections.
This last condition is added to encompass the
subtyping defined by Chen et al.~\cite{MariangiolaPreciness}
that, as discussed above, requires all messages to be eventually received:
in fact, if consecutive infinite output selections are not allowed,
it is not possible to indefinitely delay inputs.

For instance, considering the simple on-line shop example, we have:
$$T_{client} \ \not\!\!\!\!\selsubtype\ S_{client}$$
because $S_{client}$ has consecutive infinite output selections;
and $$T_{service} \ \not\!\!\!\!\selsubtype\ S_{service}$$
because $S_{service}$ has input branchings with more than one choice.

If we consider a different behavior for the shopping service, where each input branching has single choice 
$$S'_{service} = \& \{\text{add\_to\_cart}: \& \{\text{pay}: \Tend\} \}.$$
we have, instead, $T_{service} \,\selsubtype \, S'_{service} $.

The proof of undecidability of $\selsubtype$ is by reduction from 
the acceptance problem in queue machines. Queue machines are a
Turing powerful computational model composed of a finite control that consumes 
and introduces symbols in a queue, i.e. a First-In First-Out (FIFO) structure.
The input to a queue machine is given by a sequence of symbols,
ended by a special delimiter $\$$, that is initially present in the queue.
The finite control is defined by a finite set of states (one of which being 
the initial state) and a transition function that given the current state and 
the consumed symbol, i.e. the one taken from the beginning of the queue, 
returns the next state and a sequence of symbols to be added at the end 
of the queue. The input is accepted by the queue machine whenever queue 
is emptied.
As an example, one can define a queue machine able to accept the
strings $a^n b^n$, with $n \geq 0$, by considering a finite control with
the following states: 
\begin{itemize}
\item
an initial state that expects to consume one among two possible symbols:
an $a$ and then move to the second state, or 
the delimiter $\$$ (thus accepting);
\item
the second state that cyclically consumes each of the remaining symbols $a$,
re-introducing them at the end of the queue, 
and then moves to the third state by consuming 
the first $b$;
\item
the third state that cyclically consumes each of the remaining symbols $b$,
re-introducing them at the end of the queue, 
and then returns to the initial
state by consuming the $\$$ and re-enqueueing it.
\end{itemize}
Even if queue machines are similar to pushdown-automata, they are
strictly more powerful. For instance, it is trivial to extend
the above queue machine with an additional state in order to accept
strings  
$a^n b^n c^n$, with $n \geq 0$. Queue machines are not only more expressive than
pushdown-automata, but they are Turing complete. Intuitively, 
this follows from the fact that by using a queue instead of a stack, it 
is possible to access any symbol in the structure without losing the symbols 
in front. In fact, it is sufficient to re-introduce such symbols at the end 
of the queue, as done in the above example e.g. with the symbols $a$ and $b$  
in the second and the third state,
respectively. This mechanism makes it possible to simulate the tape of a Turing
machine by using the queue.

Being Turing powerful, acceptance of a string 
is undecidable for queue machines. We show (Theorem \ref{thm}) that given
a queue machine and an input string, it is alway possible to define
two session types $T$ and $S$ such that $T \,\selsubtype\, S$ if and only
if the given input string is not accepted by the considered queue machine.
From this we 
conclude that the $\selsubtype$ relation is also undecidable.

This core undecidability result allows us to prove by reduction the
undecidability of $\subtype$ as well as other more complex relations
including the three asynchronous subtypings in the literature discussed above.
Namely, we prove the undecidability of the following subtypings:
$\subtype_s$ that includes also send and receive actions and
corresponds with (a fragment of) the subtyping defined
by Mostrous and Yoshida~\cite{IandCmostrous},
$\subtype_o$ that disallows orphan messages and coincides
with the subtyping defined by Chen et al.~\cite{MariangiolaPreciness},
and $\subtype_m$ that deals with multiparty session types
and corresponds with the subtyping introduced by Mostrous et
al.~\cite{ESOP09}.

As an additional result, we show that 
further restrictions on the branching/selection structure of types 
  make our core subtyping relation $\selsubtype$ decidable. 
In fact,
  by imposing any of two possible restrictions on $\selsubtype$ --- namely, in both subtype 
  and supertype all input branchings (or all output selections) have one choice only --- 
  the obtained relation turns out to be decidable. We thus define the subtyping
  relations $\selsubtype_\mathsf{sin}$ (both types are single-choice on inputs)
  and $\selsubtype_\mathsf{sout}$ (both types are single-choice on outputs) by considering the two above restrictions on the asynchronous \selection 
  relation $\selsubtype$ and prove that both $\selsubtype_\mathsf{sin}$ and $\selsubtype_\mathsf{sout}$  
 are decidable.
  As a matter of fact, we prove decidability for larger relations w.r.t. $\selsubtype_\mathsf{sin}$
  and $\selsubtype_\mathsf{sout}$ where we do not impose the constraint about
  no consecutive infinite outputs.





\begin{figure}[t]
  \begin{center}
    \begin{tabular}{ll}
      \begin{tabular}{lllll}
        $\subtype_s$  & standard binary \cite{IandCmostrous}\\
        $\subtype_o$  & orphan-message-free \cite{MariangiolaPreciness}\\
        $\subtype_m$  & multiparty \cite{ESOP09}\\        
        $\longrightarrow$ & reduces to\\[3cm]
      \end{tabular}\hspace{-3cm}
      &
        \begin{tabular}{ccccccc}
          \xymatrixcolsep{0pc}
          \xymatrix{
          & & &&\subtype_m      \\
          & & &&\subtype_s \ar@{->}[u] &  & \subtype_o      \\
          &&&&\leq \ar@{->}[u]         \\
          &&&&              & \selsubtype \ar@{->}[ul]\ar@{->}[uur]\\
          \ar@{-}[rrrrrrrrr]^{\hspace{-5cm}\text{undecidable}}_{\hspace{-5cm}\text{decidable}}
          &&& \hspace{4cm}&&&&&& \\
          &&&&\selsubtype_\mathsf{sin} \ar@{->}[uur]& & \selsubtype_\mathsf{sout} \ar@{->}[uul]
                                             }
        \end{tabular}
    \end{tabular}
  \end{center}
  \caption{Lattice of the asynchronous subtyping relations considered in this paper.}
  \label{fig:lattice}
\end{figure}
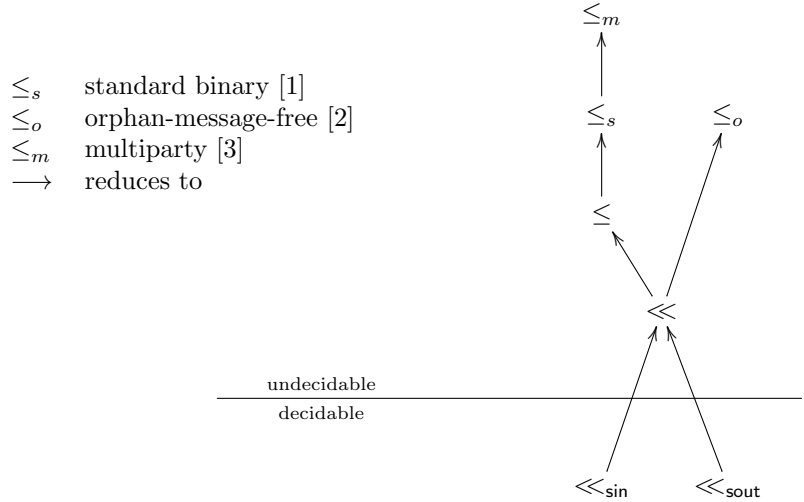

Figure \ref{fig:lattice} depicts the relations discussed
in this paper as a lattice representing a $\preceq_1 \, \longrightarrow \, \preceq_2$ order. $\preceq_1 \, \longrightarrow \,\preceq_2$ 
means that it is possible to algorithmically reduce
the problem of deciding the relation $\preceq_1$ into the problem of deciding $\preceq_2$.
As discussed above, $\subtype_s$, $\subtype_o$ and $\subtype_m$ are taken from the
literature, while $\subtype$, $\selsubtype$, $\selsubtype_\mathsf{sin}$ and $\selsubtype_\mathsf{sout}$
are defined in this paper 
to characterize as tightly as possible 
the boundary between decidability 
and undecidability for asynchronous session subtyping relations. Obviously, when a relation 
is undecidable 
all relations above it (it reduces to) are also undecidable,
while when a relation turns out to be decidable 
all relations below it (that reduce to it) are decidable as well.

\paragraph{Structure of the paper} 

In \S\ref{sec:subtyping} we introduce
a core language of session types with only branching/selection and recursion,
and define for it the asynchronous subtyping
relation $\subtype$.
In \S\ref{sec:undecidable}
we restrict subtyping to $\selsubtype$
and we show that  
such relation is undecidable.
In \S\ref{subsec:impact}, we discuss how our undecidability result 
allows us to prove the undecidability of other asynchronous 
subtypings, namely $\subtype$, $\subtype_s$, $\subtype_o$,
$\subtype_m$,
and subtyping for communicating finite state machines
(CFSMs).
In \S\ref{sec:decidable} we discuss the two decidable relations $\selsubtype_{\mathsf{sin}}$
and $\selsubtype_{\mathsf{sout}}$, obtained as further restrictions of $\selsubtype$.
Finally, in \S\ref{sec:concl}
we comment the related literature and draw some concluding remarks.

%

%% file: subtyping.tex
In this section, we give a definition of a core session type language
and define the asynchronous subtyping relation $\subtype$ following
the approach by Mostrous and Yoshida~\cite{IandCmostrous}.

\subsection{Session Types} 
We start by presenting a very simple session type language (with only
branching/selection and recursion) which is sufficient to prove our undecibility
result.
 
\begin{definition}[Session types]\label{def:sessiontypes}
  Given a set of labels $L$, ranged over by $l$, the syntax of binary
  session types is given by the following grammar:
\begin{displaymath}
  \begin{array}{lrl}
    T &::=&   
              \Tselect{l}{T} 
              \mid \Tbranch{l}{T} \mid
               \Trec t.T
               \mid \Tvar t
               \mid \Tend
  \end{array}
\end{displaymath}
\end{definition}
In our session type language we simply consider session termination
$\Tend$, recursive definitions $\Trec t.T$, with $\Tvar t$ being the
recursion variable, output selection $\Tselect{l}{T}$ 
and input branching $\Tbranch{l}{T}$. 
Each possible choice 
is labeled by a label
$l_i$, taken from the set of labels $L$, followed by a session
continuation $T_i$. Labels in a branching/selection are pairwise distinct.



\subsection{Subtyping}
We consider a notion of asynchronous subtyping corresponding to the
subtyping relation by Mostrous and Yoshida~\cite{IandCmostrous}
applied to our language. 
In particular we formalize the property of {\it output anticipation} (which, as we will see, characterizes asynchronous subtyping) by using the notion of {\it input context} as in Chen et
al.~\cite{MariangiolaPreciness} and Mostrous and
Yoshida~\cite{IandCmostrous}.
In order to define subtyping, we first
need to introduce $n$-unfolding and input contexts.

The $n$-unfolding function unfolds nested recursive definitions to
depth $n$.
\begin{definition}[$n$-unfolding]\label{def:unfolding}
  \begin{displaymath}
    \begin{array}{l@{\qquad}l}
      \unfold 0T = T 
      \\
      \unfold 1{\Tselect{l}{T}} =  \oplus\{{l}_{i}:{\unfold 1 {T_i}}\}_{i\in I}
      \\
      \unfold 1{\Tbranch{l}{T} } = \&\{{l}_{i}:{\unfold 1 {T_i}}\}_{i\in I}

      \\
      \unfold 1{\Trec t.T} = T\{\Trec t.T/\Tvar t\}
      \\
      \unfold 1{\Tvar t} = \Tvar t
      \\
      \unfold 1{\Tend} = \Tend
      \\
      \unfold n{T} = \unfold 1{\unfold {n-1}T}
    \end{array}
  \end{displaymath}
\end{definition}

\begin{definition}[Input Context]\label{def:context}
  An input context $\mathcal A$ is a session type with multiple holes
  defined by the following syntax:
  \[
  \mathcal A\ ::=\ [\,]^n \quad\mid\quad 
  \Tbranch{l}{\mathcal A}
  \]
  An input context $\mathcal A$ is well-formed whenever all its holes
  $[\,]^n$, with $n \in \mathbb{N}^+$, are consistently enumerated, i.e.\ there exists $m \geq 1$ such that $\mathcal A$ includes one and only one $[\,]^n$ for each $n \leq m$. 
  Given a well-formed input context $\mathcal A$ with holes indexed 
  over $\{1,\ldots,m\}$ and types $T_1$,\dots, $T_k$, we use $\context {A} {T_k} {k\in
        \{1,\ldots, m\}}$ to denote the type obtained by filling
  each hole $k$ in $\mathcal A$ with the corresponding term $T_k$.
\end{definition}

From now on, whenever using contexts we will assume them to be
well-formed.

For example, consider the input context 
$$\mathcal A=\&\{l_1: []^1,\ l_2: []^2\}$$
we have:
$$\context {A} {\oplus\{l:T_i \}} {i\in
        \{1,2\}} = \&\!\big\{l_1:\oplus\{l:T_1 \}, l_2:\oplus\{l:T_2 \} \big\}$$

\smallskip

We are finally ready to define our notion of subtyping.
\begin{definition}[Asynchronous Subtyping,
  $\subtype$]\label{def:subtyping}
  \label{subtype}
  $\mathcal R$ is a subtyping relation whenever $(T,S)\in\mathcal R$
  implies that:
  \begin{enumerate}


  \item 
  
    if $T=\Tend$ then $\exists n\geq 0$ such that
    $\unfold nS = \Tend$;

  \item 
  
    if $T=\Tselect{l}{T}$ then $\exists n\geq 0,\mathcal A$ such that
    \begin{itemize}
    \item
      $\unfold nS = \context {A} {\Tselectindex l{S_k}{j}{J_k}} {k\in
        \{1,\ldots, m\}}$ \qquad for some $J_k$,
    \item $\forall k\in\{1,\ldots, m\}.  I\subseteq J_k$ and
    \item
      $\forall i\in I,
      (T_i,\context {A} {{S_{ki}}} {k\in \{1,\ldots, m\}})\in\mathcal R$;
    \end{itemize}

  \item 
  
  if $T=\Tbranch{l}{T}$ then $\exists n\geq 0$ such that
    $\unfold nS = \Tbranchindex lSjJ$, $J\subseteq I$ and
    $\forall j\in J. (T_j,S_j)\in\mathcal R$;

  \item 
  
  if $T= \Trec t.{T'} $ then $(T'\{T/\Tvar{t}\}, S)\in\mathcal R$.

  \end{enumerate}
  We say that $T$ is a subtype of $S$, written $T \,\subtype \, S$, if there 
  is a subtyping relation $\mathcal R$ such that $(T,S) \in \mathcal R$.
\end{definition}


An important characteristic of asynchronous subtyping (formalized by rule~$2.$ above) is the following one. In a subtype output selections can be anticipated so to bring them before the input branchings that in the supertype occur in front of them. 
For example the type 
$$S = \&\!\big\{l_{\mathbf{but1}}:\oplus\{l_{\mathbf{coffee}}:T_1 \}, l_{\mathbf{but2}}:\oplus\{l_{\mathbf{coffee}}:T_2 \} \big\}$$  
has the following subtype
$$ 
T = \oplus\big\{l_{\mathbf{coffee}}:\&\!\{l_{\mathbf{but1}}:T_1, l_{\mathbf{but2}}:T_2 \} \big\}$$ 
where the output selection with label $l_{\mathbf{coffee}}$ is anticipated w.r.t.\ the input branching with labels $l_{\mathbf{but1}}$ and $l_{\mathbf{but2}}$. That is, since in the supertype the input branching with labels $l_{\mathbf{but1}}$ and $l_{\mathbf{but2}}$ occurs in front of the output selection with label $l_{\mathbf{coffee}}$ (which is present in {\it all} its input branches), such an output selection can be anticipated so to bring it before the $l_{\mathbf{but1}}$/$l_{\mathbf{but2}}$ input branching. 

It is, thus, immediate to verify that, according to Definition \ref{def:subtyping}, we have \mbox{$T \leq S$}.
In particular, in rule~$2.$ the well-formed input context considered to express the output $l$ anticipation is
$\mathcal A=\&\{l_{\mathbf{but1}}: []^1,\ l_{\mathbf{but2}}: []^2\}$. By considering this context, the supertype $S$ can be written as 
$$\context {A} {\oplus\{l_{\mathbf{coffee}}:T_1 \}} {i\in
        \{1,2\}} = \&\!\big\{l_{\mathbf{but1}}:\oplus\{l_{\mathbf{coffee}}:T_1 \}, l_{\mathbf{but2}}:\oplus\{l_{\mathbf{coffee}}:T_2 \} \big\}$$

Notice that in general an output selection can be anticipated even if it occurs in a larger input context, such as, for example
$$\mathcal A=\&\{l_{\mathbf{but1}}: \&\{l_{\mathbf{but3}}:[]^1\},\ l_{\mathbf{but2}}: []^2\}$$
%

Conceptually output anticipation
reflects the fact that we are considering asynchronous communication 
protocols in which messages are stored in queues. 
In this setting, it is safe to replace
a peer that follows a given protocol with another
one following a modified protocol where
outputs are anticipated: in fact, the difference 
is simply that such outputs 
will be stored earlier in the communication queue.
%
%
%
%

\subsection{Examples}
\label{sec:examples}

Consider the types
$$\begin{array}{l}
T=\Trec t.\Tbranchsimple{l}{\Tselectsimple{l}{\Tvar t}}\\
S=\Trec t.\Tbranchsimple{l}{\Tbranchsimple{l}{\Tselectsimple{l}{\Tvar
      t}}}
\end{array}$$
We have 
$T \, \subtype \, S$ as the following infinite set of type pairs
is a subtyping relation:
$$
\begin{array}{ll}
  \{ & \big(T\ ,\ S\big), 
     \big(\Tbranchsimple{l}{\Tselectsimple{l}{T}}\ ,\ S\big),\ 
     \big(\Tselectsimple{l}{T}\ ,\ \Tbranchsimple{l}{\Tselectsimple{l}{S}}\big),\ \\
     & \big(T\ ,\ \Tbranchsimple{l}{S}\big),\ 
     \big(\Tbranchsimple{l}{\Tselectsimple{l}{T}}\ ,\ \Tbranchsimple{l}{S}\big),\ 
     \big(\Tselectsimple{l}{T}\ ,\ S\big),\\
     & \big(T\ ,\ \Tbranchsimple{l}{\Tbranchsimple{l}{S}}\big),\ 
     \dots\\
     & \big(T\ ,\ \Tbranchsimple{l}{\Tbranchsimple{l}{\Tbranchsimple{l}{S}}}\big),\ \dots\\
     & \dots\ \ \ \}
\end{array}
$$
Notice that the types on the r.h.s. ($S$ and subsequent ones) can
always mimic the initial actions of the corresponding type on the l.h.s. ($T$ and
subsequent ones). Pairs are presented above in such a way that: the second one is reached from 
the first one by rule $4.$\ of Definition \ref{def:subtyping} (recursion), the third one is reached from 
the second one by rule $3.$\ of Definition \ref{def:subtyping} (input), the first one in the subsequent line is
reached from the third one by rule $2.$\ of Definition \ref{def:subtyping} (output), and similarly (using the same rules) in the subsequent lines.
Notice that, every time an output
${\Tselectsimple{l}{\_}}$ must be mimicked, the r.h.s. must be
unfolded, and the corresponding output is anticipated, since
it is preceded by inputs only (i.e. the output fills an input context). 
The effect of the anticipation of the output is that
a new input $\Tbranchsimple{l}{\_}$ is accumulated at the beginning
of the r.h.s.  It is worth to observe that
every accumulated input $\Tbranchsimple{l}{\_}$ is
eventually consumed in the simulation game, but 
the accumulated inputs grows unboundedly.

As another example consider
$$\begin{array}{lll}
T= \Trec t.\& \!\!\! & \!\!\! \{\\
&\hspace{.1cm}l_{\mathbf{but1}}: \oplus\{l_{\mathbf{coffee}}:\Tvar t\}, \hspace{-3cm}&\\
&\hspace{.1cm}l_{\mathbf{but2}}: \oplus\{l_{\mathbf{tea}}:\Tvar t\} \hspace{-3cm}& \\
&\!\!\!\}&\\[.1cm]
S= \Trec t.\& \!\!\! & \!\!\! \{ &\\
&\hspace{.1cm}l_\mathbf{but2}: \oplus \!\!\! & \!\!\! \{\\
& &\hspace{.1cm}l_\mathbf{coffee}:\Tvar t,\\
& &\hspace{.1cm}l_\mathbf{tea}:\&\{l_\mathbf{but1}:\Tvar t,\ l_\mathbf{but2}:\Tvar t\}\\
& & \!\!\!\}\\
&\!\!\!\} &
\end{array}$$
We have $T \, \subtype \, S$ for the following reasons.  Type $T$ repeatedly
alternates input and output, the input corresponding to an input branching with labels $l_\mathbf{but1}$ and $l_\mathbf{but2}$ (where $\mathbf{but}$ stands for ``button''), and the output
with only one label: either $l_\mathbf{coffee}$ or $l_\mathbf{tea}$. Also $S$
infinitely repeats input and output, but, depending on which of its inputs it performs, the corresponding input branching can have fewer
choices than $T$ (in the case of the input branching with just one label, i.e.\ $l_\mathbf{but2}$). The
output, instead, always corresponds to an output selection with labels $l_\mathbf{coffee}$ and $l_\mathbf{tea}$.  Such a difference between $T$ and $S$ is not problematic due
to contravariance on input branchings and covariance on output selections.
Type $S$ also differs because after the input with label $l_\mathbf{but2}$ and the
output with label $l_\mathbf{tea}$ the type $\&\{l_\mathbf{but1}: S,\ l_\mathbf{but2}: S\}$ is reached,
and this term (no matter which input is chosen) may perform two consecutive inputs, the second one available upon
unfolding of $S$.  Type $T$ does not have two consecutive inputs
because it always alternates input and output.  Nevertheless, we still
have $T \, \subtype \; \&\{l_\mathbf{but1}: S,\ l_\mathbf{but2}: S\}$ because, as discussed in the
previous example, according to our notion of asynchronous subtyping (rule $2.$\ of Definition \ref{def:subtyping})
the output in the r.h.s. can be always anticipated to match the output
actions of the l.h.s., if they are preceded by inputs only.

%% file: undecidable.tex
%
In this section we prove our core undecidability result for a restricted subtyping relation.
This relation is called 
asynchronous \selection subtyping and corresponds to the
restriction of $\subtype$ to pairs of session types $(T,S)$
such that $T$ has output selections with one choice only,
$S$ has input branchings with one choice only, and both
$T$ and $S$ cannot have infinite sequences of outputs.



\subsection{Asynchronous single-choice subtyping}
In order to formally define such relation we need some preliminary
definitions. 
\begin{definition}[Session types with single
  outputs]\label{def:sessiontypesnooutput}
  Given a set of labels $L$, ranged over by $l$, the syntax of binary
  {\em session types with single outputs} is given by the following
  grammar:
  \begin{displaymath}
    \begin{array}{lrl}
      T^{\textsf{out}} &::=&   
                             \Tselectsingle{l}{T^{\textsf{out}}} 
                             \mid \Tbranch{l}{T^{\textsf{out}}} \mid
                             \Trec t.{T^{\textsf{out}}}
                             \mid \Tvar t
                             \mid \Tend
    \end{array}
  \end{displaymath}
\end{definition}

Session types with single outputs are all those session types where
inputs can have multiple choices while outputs must be singletons.

\begin{definition}[Session types with single
  inputs]\label{def:sessiontypesnoinput}
  Given a set of labels $L$, ranged over by $l$, the syntax of binary
  {\em session types with single inputs} is given by the following
  grammar:
  \begin{displaymath}
    \begin{array}{lrl}
      T^{\textsf{in}} &::=&   
                             \Tselect{l}{T^{\textsf{in}}} 
                             \mid \Tbranchsingle{l}{T^{\textsf{in}}} \mid
                             \Trec t.{T^{\textsf{in}}}
                             \mid \Tvar t
                             \mid \Tend
    \end{array}
  \end{displaymath}
\end{definition}

Session types with single inputs, instead, are all those session types
where inputs are singletons and outputs may have multiple choices.

\begin{definition}[Session types with input guarded recursion]\label{def:sessiontypesnoinfinite}
  The set of session types $\noinf$ is composed of the session types $T$
  that satisfy the following condition: for every subterm of $T$ of the form $\Trec t.R$, 
  for some $\Tvar t$ and $R$, every occurrence of $\Tvar t$ in $R$ is inside a subterterm 
  of $R$ of the form $\Tbranch{l}{S}$, for some set of labels $l_i$ and session types $S_i$.
\end{definition}

Session types with input guarded recursion have an important property:
there are no consecutive infinite outputs.


The asynchronous \selection relation is defined as the subset of
$\subtype$ where types on the left-hand side of the relation are with
single outputs, 
types on the right-hand side are with single inputs,
and both types have no infinite sequences of outputs.

\begin{definition}[Asynchronous \Selection
  Relation]\label{def:selection_subtyping}
  The asynchronous \selection relation $\selsubtype$ is defined as:
  \[ \selsubtype\ =\  \subtype\cap (T^{\textsf{out}} \times T^{\textsf{in}})
  \cap (\noinf \times \noinf)\]
\end{definition}

\paragraph{Remark}
Note that the asynchronous \selection relation is not reflexive. In
fact, any type that has multiple (non single) choices is not related to itself, e.g.,
$\&\{l_1: \Tend,\ l_2: \Tend\} \ \not\!\!\!\!\selsubtype\ \&\{l_1:
\Tend,\ l_2: \Tend\}$
simply because the term on right-hand side has more than one input branch.

\subsection{Queue Machines}\label{subsec:queuemachines}
The proof of undecidability of the asynchronous \selection relation is
by reduction from the acceptance problem for queue machines. Queue
machines have been already informally presented in the Introduction,
we now report their formal definition.
%
\begin{definition}[Queue machine]\label{def:queuemachines}
  A queue machine $M$ is defined by a six-tuple
  $(Q , \Sigma , \Gamma , \$ , s , \delta )$ where:
  \begin{itemize}
  \item $Q$ is a finite set of states;
  \item $\Sigma \subset \Gamma$ is a finite set denoting the input
    alphabet;
  \item $\Gamma$ is a finite set denoting the queue alphabet (ranged
    over by $A,B,C,X$);
  \item $\$ \in \Gamma -\Sigma$ is the initial queue symbol;
  \item $s \in Q$ is the start state;
  \item $\delta : Q \times \Gamma \rightarrow Q\times \Gamma ^{*}$ is
    the transition function.
  \end{itemize}
\end{definition}

In the Introduction we have informally described a queue machine
that accepts the language $a^n b^n$, we now present its formal definition.
Let $M=(\{q_1,q_2,q_3,q_s\},\{a,b\},\{a,b,\$\},\$ , q_1 , \delta )$  
with $\delta$ defined as follows:
\begin{itemize}
\item
$\delta(q_1,a) = (q_2,\epsilon)$,
$\delta(q_1,\$) = (q_1,\epsilon)$,
$\delta(q_1,b) = (q_s,b)$;
\item
$\delta(q_2,a) = (q_2,a)$,
$\delta(q_2,\$) = (q_s,\$)$,
$\delta(q_2,b) = (q_3,\epsilon)$;
\item
$\delta(q_3,a) = (q_s,a)$,
$\delta(q_3,\$) = (q_1,\$)$,
$\delta(q_3,b) = (q_3,b)$;
\item
$\delta(q_s,X) = (q_s,X)$ for $X \in \{a,b,\$\}$.
\end{itemize}
Differently from the informal definition in the Introduction,
here (since, according to Definition \ref{def:queuemachines}, the transition function $\delta$ is expected to be total) we have to consider an additional sink state which is 
entered whenever an unexpected symbol is consumed from the queue.
Once this state is entered, it will be no longer possible to leave it,
and every consumed symbol will be simply re-added to the queue.

We now formally define queue machine computations.
  
\begin{definition}[Queue machine computation]  
  A {\em configuration} of a queue machine is an ordered pair
  $(q,\gamma)$ where $q\in Q$ is its {\em current state} and
  $\gamma\in\Gamma ^{*}$ is the 
  {\em queue} ($\Gamma ^{*}$ is the Kleene closure of $\Gamma$).  The
  starting configuration on an input string $x$ is $(s , x \$)$.  The
  transition relation $\rightarrow _{M}$ over configurations
  $Q \times \Gamma ^{*}$, leading from a configuration to the
  next one, is defined as follows. For any $p,q \in Q$, $A \in \Gamma$
  and $\alpha,\gamma \in \Gamma ^{*}$ we have
  $(p,A\alpha )\rightarrow _{M}(q,\alpha \gamma)$ whenever
  $\delta (p,A)=(q,\gamma)$.  A machine $M$ accepts an input $x$ if it
  eventually terminates on input $x$, i.e. it reaches a blocking
  configuration with the empty queue (notice that, as the transition
  relation is total, the unique way to terminate is by emptying the
  queue).  Formally, $x$ is accepted by $M$ if
  $(s,x\$)\rightarrow _{M}^{*}(q,\epsilon)$ where $\epsilon$ is the
  empty string and $\rightarrow _{M}^{*}$ is the reflexive and
  transitive closure of $\rightarrow _{M}$.
\end{definition}

Going back to the queue machine $M$ defined above, 
if we consider the input $aabb$ we have the following computation:
$$
\begin{array}{l}
(q_1,aabb\$) \rightarrow _{M}
 (q_2,abb\$) \rightarrow _{M}
 (q_2,bb\$a) \rightarrow _{M}
 (q_3,b\$a) \rightarrow _{M}
 (q_3,\$ab) \rightarrow _{M} \\
 (q_1,ab\$) \rightarrow _{M}
 (q_2,b\$) \rightarrow _{M}
 (q_3,\$) \rightarrow _{M}
 (q_1,\$) \rightarrow _{M}
 (q_1,\epsilon) 
\end{array}
$$
Hence, we can conclude that the string $aabb$ is accepted by $M$
(as any other string of type $a^n b^n$).

Turing completeness of queue machines is discussed by
Kozen~\cite{KozenBook} (page~354, solution to exercise~99). A
configuration of a Turing machine (tape, current head position and
internal state) can be encoded in a queue, and a queue machine can
simulate each move of the Turing machine by repeatedly consuming and
reproducing the queue contents, only changing the part affected by the
move itself.  
Formally, given any Turing machine $T$ we have that a string $x$ is accepted by $T$ if and only if 
$x$ is accepted by the queuing machine $M$ obtained as the encoding of $T$.
The undecidability of acceptance of an input string $x$
by a machine $M$ follows directly from such encoding.

\subsection{Modelling 
  Queue Machines with Session Types}
\label{subsec:encoding}
Our goal is to construct a pair of types, say $T$ and $S$, from a
given queue machine $M$ and a given input $x$, such that: $T \,\selsubtype\, S$ if and only if $x$
is not accepted by $M$.
%
%
Intuitively, type $T$ encodes the finite control of $M$, i.e., its
transition function $\delta$, starting from its initial state
$s$. And type $S$ encodes the machine queue that initially contains
$x \$$, where $x$ is the input string $x=X_1 \cdots X_n$ of length $n
\geq 0$.  The set of labels $L$ for such types $T$ and $S$ is
$M$'s queue alphabet~$\Gamma$.

Formally, the queue of a machine is encoded into a session type as
follows:
\begin{definition}[Queue Encoding]\label{def:queueenc}
  Let $M=
  (Q , \Sigma , \Gamma , \$ , s , \delta
  )$ be a queue machine and let $C_1 \cdots C_m \in \Gamma^*$, with $m
  \geq 0$.  Then, the queue encoding function $\semS{C_1 \cdots
    C_m}$ is defined as:
  $$
  \begin{array}{l}
    \semS{C_1\!\cdots\! C_m} =
    \Tbranchsingle{C_1\!\!}{
    \!\ldots
    \Tbranchsingle{C_m}
    {\Trec{t}.\Tselectset{A}{\Tbranchsingle{A}{\Tvar{t}}}{\Gamma}}
    }
  \end{array}
  $$
  Given a configuration $(q,\gamma)$
  of $M$,
  the encoding of the queue $\gamma=C_1
  \cdots C_m$ is thus defined as $\semS{C_1 \cdots C_m}$.
\end{definition} 
Note that whenever $m=0$,
we have
$\semS{\epsilon}={\Trec{t}.\Tselectset{A}{\Tbranchsingle{A}{\Tvar{t}}}{\Gamma}}$.
Observe that we are using a slight abuse of notation:
in both output selections and input branchings, labels $l_A$, with $A\in \Gamma$, are simply denoted by $A$.

\begin{figure}[!t]
  \centering
  \includegraphics[width=3in]{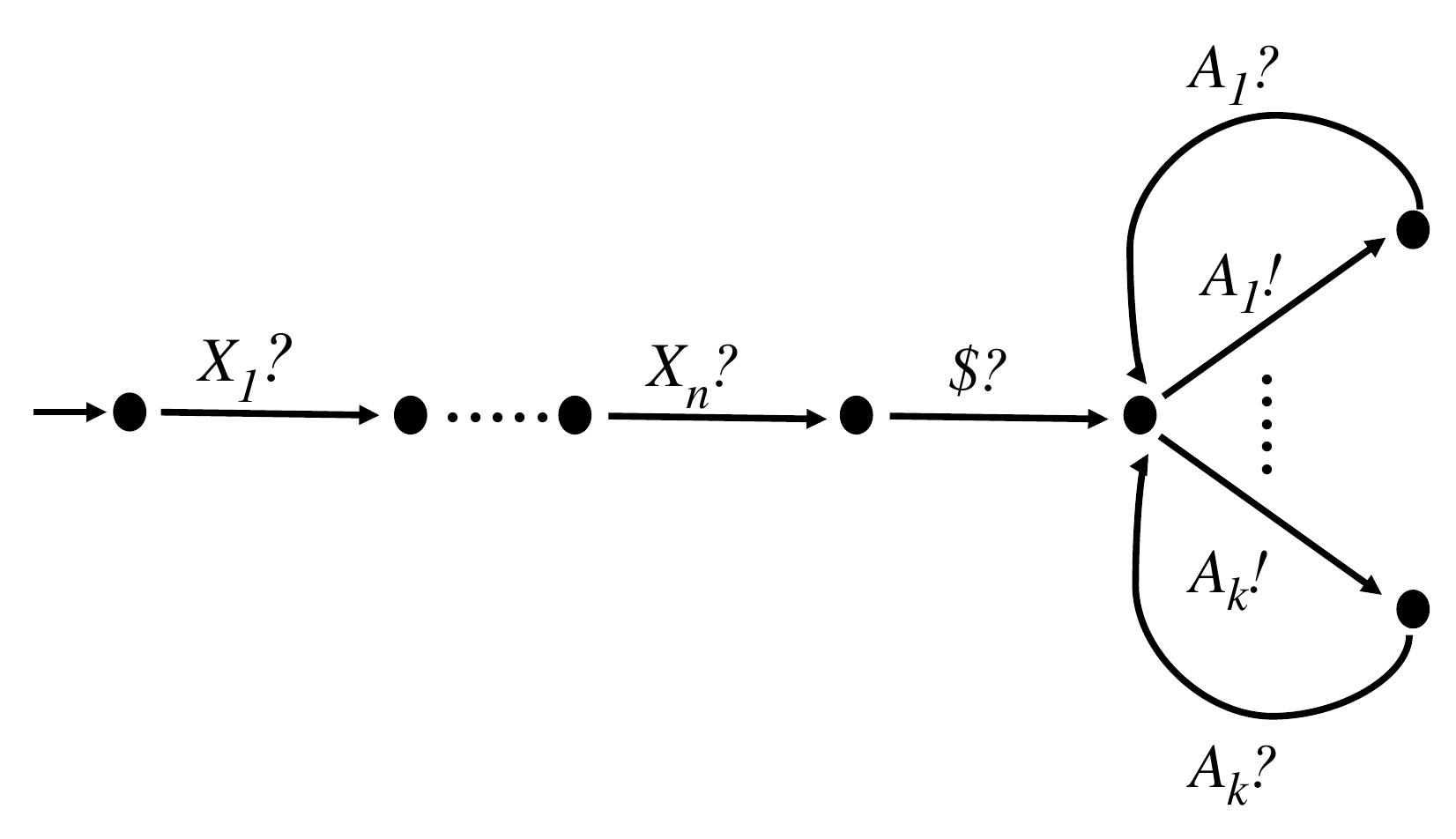}
  \caption{Labelled Transition System of a session type 
    encoding 
    the initial queue $X_1\cdots X_n\$$
  }
  \label{fig:queue}
\end{figure}
Figure~\ref{fig:queue} contains a graphical representation of 
the queue encoding
with its initial content $X_1\cdots X_n\$$.
In order to better clarify our development, we graphically represent
session types as labeled transition systems (in the form of
communicating automata \cite{BZ83}), where an output selection
$\Tselect{l}{T}$
is represented as a choice among alternative output transitions
labeled with ``$l_i
!$'', and an input branching
$\Tbranch{l}{T}$
is represented as a choice among alternative input transitions labeled
with ``$l_i ?$''.
%
Intuitively, we encode a queue containing symbols $C_1
\cdots C_m$ with a session type that starts with
$m$
inputs with labels $C_1$,
$\dots$,
$C_m$,
respectively. Thus, in Figure~\ref{fig:queue}, we have $C_1
\cdots C_m = X_1 \cdots X_n \$ $. 
After such sequence of inputs, representing the current queue content,
there is a recursive type representing the capability to enqueue new
symbols. Such a type repeatedly performs an output selection with one choice
for each symbol $A_i$
in the queue alphabet $\Gamma$
(with $k$
being the cardinality of $\Gamma$),
followed by an input labeled with the same symbol $A_i$.

\smallbreak

We now give the definition of the type modelling the finite control of
a queue machine, i.e., the encoding of the transition function
$\delta$.  
\begin{definition}[Finite Control Encoding]\label{def:controlEncoding}
  Let $M= (Q , \Sigma , \Gamma , \$ , s , \delta )$ be a queue machine
  and let $q\in Q$ and $\mathcal S\subseteq Q$. Then,
  \begin{displaymath}
    \begin{array}{l}
      \semT{q}{\mathcal S} = 
      \left \{
      \begin{array}{l}
          \Trec{q}. 
          \Tbranchset{A}{
          \Tselectsingle{B^A_1}{\cdots
          \Tselectsingle{B^A_{n_A}}{\semT{q'}{\mathcal S \cup q}
          }
          }
          }                  
          {\Gamma} 
        \\[1mm]
        \hspace{0.9cm}\text{if }q\not\in {\mathcal S} \text{ and } \delta(q,A)=(q',B^A_1\cdots B^A_{n_A})
        \\
        \\
        \Tvar{q}\qquad \mbox{if $q \in {\mathcal S}$}
      \end{array}
      \right.
    \end{array}
  \end{displaymath}
  The encoding of the transition function of $M$ 
  is then defined as $\semT{s}{\emptyset}$.
\end{definition}
The finite control encoding is a
recursively defined term with one recursion variable $\Tvar{q}$ for
each state $q \in Q$ of the machine.  
Above, $\semT{q}{\mathcal S}$ is a function that, given a state $q$ and
a set of states ${\mathcal S}$, returns a type representing the
possible behaviour of the queue machine starting from state $q$. Such
behaviour consists of first reading from the queue (input branching on
$A\in\Gamma$) and then writing on the queue a sequence of symbols
$B^A_1,\ldots, B^A_{n_A}$.  The parameter ${\mathcal S}$ is necessary
for managing the recursive definition of this type. In fact, as the
definition of the encoding function is itself recursive, this
parameter keeps track of the states that have been already encoded
(see example below).
%
%
%
%
\begin{figure}[!t]
  \centering
  \includegraphics[width=3in]{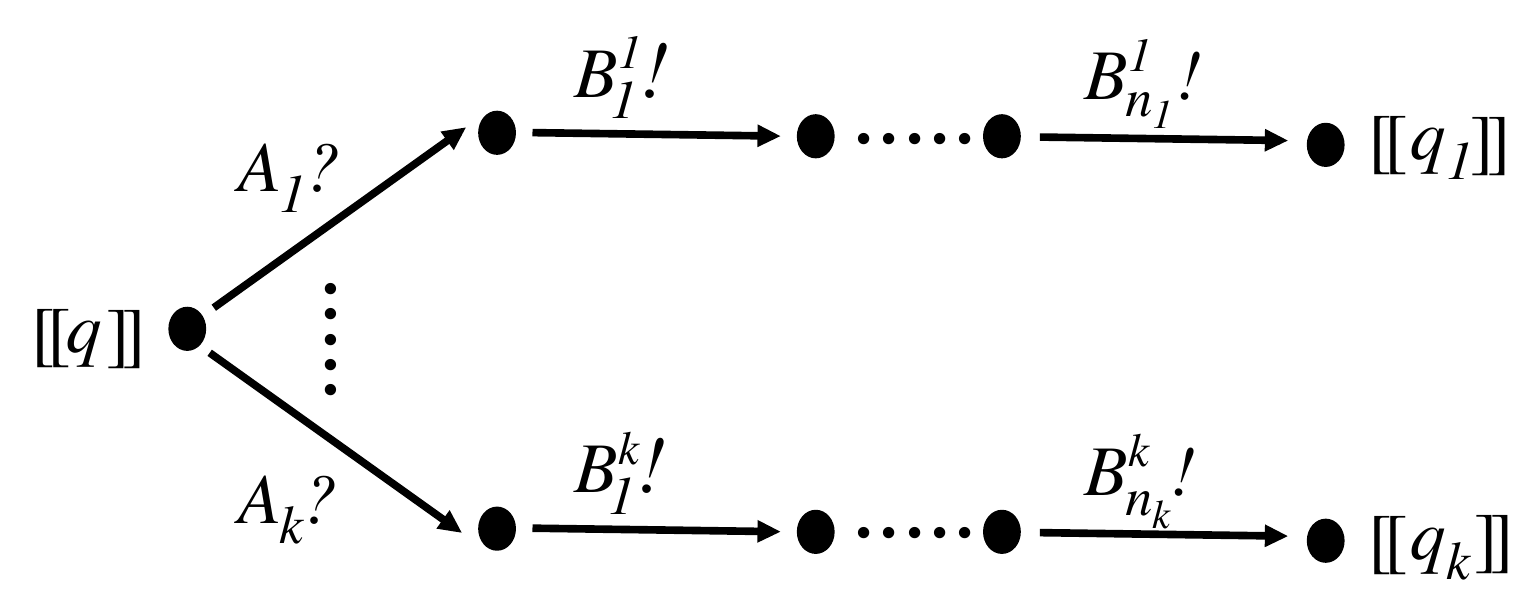}
  \smallbreak (for $\Gamma=\{A_i|i \leq k\}$ and
  $\delta(q,A_i)=(q_i,B^i_1\cdots B^i_{n_i})$ for every $i$.)
  \caption{Labelled Transition System of a session type 
    encoding a finite control.}
  \label{fig:control}
\end{figure}
In Figure~\ref{fig:control}, we report a graphical representation of
the Labelled Transition System corresponding to the session type that
encodes the queue machine finite control, i.e. the transition function
$\delta$. Each state $q \in Q$ is mapped onto a state $\semS{q}$ of a
session type, which performs an input branching with a choice for each
symbol in the queue alphabet $\Gamma$ (with $k$ being the cardinality
of $\Gamma$).  Each of these choices represents a possible character
that can be read from the queue.  After this initial input branching,
each choice continues with a sequence of outputs labeled with the
symbols that are to be inserted in the queue (after the symbol
labeling that choice has been consumed). This is done
according to function $\delta$, assuming that
$\delta(q,A_i)=(q_i, B^i_1 \cdots B^i_{n_i})$, with $n_i \geq 0$, for
all $i$ in $\{1, \dots, k\}$.  After the insertion phase, state
$\semS{q_i}$ of the session type corresponding to state $q_i$ of the
queue machine is reached.

Notice that, queue insertion actually happens in the encoding because,
when the encoding of the finite control performs an output of a $B$
symbol, the encoding of the queue must mimic such an output, possibly
by anticipating it. This has the effect of adding an input on $B$ at
the end of the sequence of initial inputs of the queue machine encoding.

Observe that our encodings generate terms that belong to the
restricted syntax of session types introduced in the previous section,
namely the queue encoding of Definition \ref{def:queueenc} produces
types in $T^{\textsf{in}}$, while the finite control encoding of
Definition \ref{def:controlEncoding} produces types in
$T^{\textsf{out}}$.

\paragraph{Example}
As an example, consider a queue machine with two states $s$ and
$q$, any non-empty input and queue alphabets $\Sigma$ and $\Gamma$, and a transition relation defined as follows: $\delta(s,A)=(q,A)$
and $\delta(q,A)=(s,\epsilon)$, for every queue symbol $A \in \Gamma$.  We have
that 
$$\begin{array}{lcl}
\semT{s}{\emptyset} & = &
\Trec{s}.  \Tbranchset{A}
{\Tselectsingle{A} {\semT{q}{\{s\}}} } {\Gamma} \\
&= &\Trec{s}.
\Tbranchset{A} {\Tselectsingle{A} {\Trec{q}.\Tbranchset{A}
    {\semT{s}{\{s,q\}}} {\Gamma}} } {\Gamma} \\
& = &\Trec{s}.
\Tbranchset{A} {\Tselectsingle{A} {\Trec{q}.\Tbranchset{A} {\Tvar s}
    {\Gamma}} } {\Gamma} 
\end{array}$$


\subsection{Properties of the Encodings}

We begin by proving that subtyping is preserved by reductions
of queue machines (modulo our encoding), and then we exploit
this property (Lemma \ref{lemmaOnlyIf}) to prove our core undecidability
result (Theorem \ref{thm}).

\begin{lemma}
\label{lemmaOnlyIf}
Consider a queue machine $M= (Q , \Sigma , \Gamma , \$ , s , \delta )$.
If $(q,\gamma) \rightarrow _{M} (q',\gamma')$
and $\semT{q}{\emptyset} \,\selsubtype \,\semS{\gamma}$ then
also $\semT{q'}{\emptyset} \,\selsubtype \,\semS{\gamma'}$.
\end{lemma}
\proof
Assume that $(q,\gamma) \rightarrow _{M} (q',\gamma')$ with
$\gamma=C_1\cdots C_m$ and $\delta(q,C_1)=(q',B^{C_1}_1\cdots B^{C_1}_{n_{C_1}})$.
Then we have that $\gamma'=C_2\cdots C_m B^{C_1}_1\cdots B^{C_1}_{n_{C_1}}$.

Assume now $\semT{q}{\emptyset} \, \selsubtype \, \semS{\gamma}$;
this means that there exists an asynchronous subtyping relation ${\mathcal R}$
s.t. $(\semT{q}{\emptyset},\semS{C_1\cdots C_m}) \in {\mathcal R}$.
By item 4 of Definition \ref{subtype}, we also have
$$
\big(           
          \Tbranchset{A}{
          \Tselectsingle{B^A_1}{\cdots
          \Tselectsingle{B^A_{n_A}}{\semT{q'}{\emptyset}
          }
          }
          }                  
          {\Gamma} 
                        \  ,\
                        \semS{C_1\cdots C_m}\ \big) \in {\mathcal R}
$$
where the l.h.s. has been unfolded once.  By item 3 of Definition
\ref{subtype}, the presence of the above pair in ${\mathcal R}$
guarantees also that
$$
\big(\Tselectsingle{B^{C_1}_1}{
         \Tselectsingle{B^{C_1}_2}{\cdots
           \Tselectsingle{B^{C_1}_{n_{C_1}}}{\semT{q'}{\emptyset}
                              }
                            }
                          }
\ ,\
 \unfold{n_0}{\Tbranchsingle{C_2}{\cdots
           \Tbranchsingle{C_m}
             {Z}}
                              }
\ \big) \in  {\mathcal R}
$$
for some $n_0$,
and with $Z = {\Trec{t}.\Tselectset{A}{\Tbranchsingle{A}{\Tvar{t}}}{\Gamma}}$.
By item 2 of Definition \ref{subtype}, the presence of this last pair
in ${\mathcal R}$ guarantees also that
$$
\big(
         \Tselectsingle{B^{C_1}_2}{\cdots
           \Tselectsingle{B^{C_1}_{n_{C_1}}}{\semT{q'}{\emptyset}
                              }
                            }                         
\ ,\
 \unfold{n_1}{\Tbranchsingle{C_2}{\cdots
           \Tbranchsingle{C_m}
             {\Tbranchsingle{B^{C_1}_1}
             {Z}}}
                              }
\ \big) \in  {\mathcal R}
$$
for some $n_1$.
By repeating the same reasoning on $B^{C_1}_2,\cdots,B^{C_1}_{n_{C_1}}$
we conclude that also 
$$
\big(\semT{q'}{\emptyset}                      
\ ,\
 \unfold{n_{C_1}}{\Tbranchsingle{C_2}{\cdots
           \Tbranchsingle{C_m}
             {\Tbranchsingle{B^{C_1}_1}
             {\cdots \Tbranchsingle{B^{C_1}_{n_{C_1}}}
             {Z}}}}
                              }
\ \big) \in  {\mathcal R}
$$
for some $n_{C_1}$.

We now observe that if an asynchronous
subtyping relation $\mathcal R$ contains the pair
$(T,\unfold{n}{S})$, for some $n$, then we have that also 
the following is an asynchronous subtyping relation:
$\mathcal R \cup \{(T',S) | T' \in \recclose{T} \}$,
where $\recclose{T}$ is the minimal set of types that
contains $T$ and such that $\Trec{t}.{R} \in \recclose{T}$
implies $R\{\Trec{t}.{R} / \Tvar t\} \in \recclose{T}$.

In the light of this last observation we can now conclude that having 
$$
\big(\semT{q'}{\emptyset}                      
\ ,\
 \unfold{n_{C_1}}{\Tbranchsingle{C_2}{\cdots
           \Tbranchsingle{C_m}
             {\Tbranchsingle{B^{C_1}_1}
             {\cdots \Tbranchsingle{B^{C_1}_{n_{C_1}}}
             {Z}}}}
                              }
\ \big)
$$
in the asynchronous subtyping relation ${\mathcal R}$
implies that 
$$
\semT{q'}{\emptyset}                      
\ \selsubtype\
 {\Tbranchsingle{C_2}{\cdots
           \Tbranchsingle{C_m}
             {\Tbranchsingle{B^{C_1}_1}
             {\cdots \Tbranchsingle{B^{C_1}_{n_{C_1}}}
             {Z}}}}
                              }
$$
Notice that the r.h.s. corresponds to 
$\semS{\gamma'}$. Hence we have proved the thesis
$\semT{q'}{\emptyset}
\, \selsubtype \, \semS{\gamma'}$.
\qed

\medskip
We are now ready to prove our main theorem.
\begin{theorem}\label{thm}
  Given a queue machine $M= (Q , \Sigma , \Gamma , \$ , s , \delta )$,
  an input string $x$, and the two types $T=\semT{s}{\emptyset}$ and
  $S=\semS{x\$}$, we have that $M$ accepts $x$ if and only if
  $T \, \not\!\!\!\selsubtype \, S$.
\end{theorem}
\proof We prove the two directions separately.

{\em (Only if part).}  We first observe that
$\semT{q}{\emptyset} \,\not\!\!\!\!\selsubtype \, \semS{\epsilon}$ for
every possible state $q$. In fact, $\semT{q}{\emptyset}$ is a
recursive definition that upon unfolding begins with an input branching that implies (according to items 3. and 4. of Definition
\ref{subtype}) that also $\semS{\epsilon}$ (once unfolded, if needed)
should start with an input branching.  But this is false, in that, by
definition of the encoding we have
$\semS{\epsilon} =
{\Trec{t}.\Tselectset{A}{\Tbranchsingle{A}{\Tvar{t}}}{\Gamma}}$.
We can conclude that
$\semT{s}{\emptyset} \not\!\!\!\!\selsubtype \, \semS{x\$}$ because
otherwise, by repeated application of Lemma \ref{lemmaOnlyIf}, we
would have that the termination of the queue machine, i.e.
$(s,x\$)\rightarrow _{M}^* (q,\epsilon)$, implies the existence of a
state $q$ such that $\semT{q}{\emptyset} \,\selsubtype\, \semS{\epsilon}$.
But we have just proved that
$\semT{q}{\emptyset} \, \selsubtype \, \semS{\epsilon}$ does not hold, for
every state $q$.

{\em (If part).}  Our aim is to show that if
$T \!\not\!\!\!\!\selsubtype \, S$ then $M$ accepts $x$, which is
equivalent to showing that $T \, \selsubtype \, S$, assuming that $M$ does
not accept $x$.
When a queue machine does not accept an input,
the corresponding computation never ends. In our case, this means that
there is an infinite
sequence $(s,x\$) \rightarrow _{M} (q_1,\gamma_1) \rightarrow _{M}
\cdots \rightarrow _{M} (q_i,\gamma_i) \rightarrow _{M} \cdots$.
Let $\mathcal C$ be the set of reachable configurations, i.e.
$\mathcal C=\{(q_i,\gamma_i)\ |\ i \geq 0\}$ where we assume $(q_0,\gamma_0)=(s,x\$)$. 
We now define a relation $\mathcal R$ on types:
$$
\begin{array}{ll}
\mathcal{R} & =
\\
\{ 
&
\big(\ \semT{q}{\emptyset}, \semS{C_1 \cdots C_m}\ \big),\\
&
\big(\ 
\Tbranchset{A}{
          \Tselectsingle{B^A_1}{\cdots
          \Tselectsingle{B^A_{n_A}}{\semT{q'}{\emptyset}
          }
          }
          }                  
          {\Gamma}\ ,\ \semS{C_1 \cdots C_m}\ \big),
\\
&
\big(\Tselectsingle{B^{C_1}_1}{
         \Tselectsingle{B^{C_1}_2}{\cdots
           \Tselectsingle{B^{C_1}_{n_{C_1}}}{\semT{q'}{\emptyset}
                              }
                            }
                          }
\ ,\
 \Tbranchsingle{C_2}{\cdots
           \Tbranchsingle{C_m}
             {Z}
                              }
\ \big),                                                        

\\

&
\big(      
         \Tselectsingle{B^{C_1}_2}{\cdots
           \Tselectsingle{B^{C_1}_{n_{C_1}}}{\semT{q'}{\emptyset}
                              }
                            }                          
\ ,\
 \Tbranchsingle{C_2}{\cdots
           \Tbranchsingle{C_m}{
           \Tbranchsingle{B^{C_1}_1}{  
             {Z}
             					    }
                              }
                              }
\ \big),                                                        

\\

& \cdots

\\
&
\big(\      
         \semT{q'}{\emptyset}
                                                      
\ ,\
 \Tbranchsingle{C_2}{\cdots
           \Tbranchsingle{C_m}{
           \Tbranchsingle{B^{C_1}_1}{  \cdots
             {           \Tbranchsingle{B^{C_1}_{n_{C_1}}}{  
             {Z}
             					    }}
             					    }
                              }
                              }
\ \big)

\\

| & (q,C_1 \cdots C_m) \in \mathcal C, \delta(q,C_1)=(q',B^{C_1}_1\cdots B^{C_1}_{n_{C_1}}), \\
&
Z = {\Trec{t}.\Tselectset{A}{\Tbranchsingle{A}{\Tvar{t}}}{\Gamma}}
\ \ \}
\end{array}
$$
Notice that the type pairs listed in the above definition correspond
to the pairs discussed in the proof of Lemma \ref{lemmaOnlyIf}.  We
have that the above $\mathcal{R}$
is a subtyping relation because, using a reasoning similar to the one
reported in the proof of Lemma \ref{lemmaOnlyIf}, it is immediate to
see that each of the pairs satisfies the conditions in Definition
\ref{subtype} thanks to the presence of the subsequent pair.  The
unique pair without a subsequent pair is the last one, but this last
pair corresponds to the first one of the pairs corresponding to the
configuration $(q',C_2
\cdots C_m B^{C_1}_1\cdots B^{C_1}_{n_{C_1}}) \in {\mathcal
  C}$ reached in the queue machine computation after $(q,C_1 \cdots
C_m)$, i.e. $(q,C_1 \cdots C_m) \rightarrow_M (q',C_2 \cdots C_m
B^{C_1}_1\cdots
B^{C_1}_{n_{C_1}})$.  We can conclude observing that $T\ \selsubtype\
S$ because $T=\semT{s}{\emptyset}$, $S=\semS{x\$}$ and $(s,x\$) \in
{\mathcal C}$ implies that $( \semT{s}{\emptyset}, \semS{x\$}
)$ belongs to the above subtyping relation ${\mathcal R}$.  \qed

\medskip As a corollary, we have that the asynchronous \selection
relation is undecidable. 
\begin{corollary}\label{cor:undecidable}
  Asynchronous \selection subtyping for binary session types
  $\selsubtype$ is undecidable.
\end{corollary}
\begin{proof}
  From Theorem \ref{thm} we know that 
  given any queue machine $M= (Q , \Sigma , \Gamma , \$ , s , \delta )$,
  an input string $x$, and the two types $T=\semT{s}{\emptyset}$ and
  $S=\semS{x\$}$, we have that $M$ does not accept $x$ if and only if
  $T \, \selsubtype \, S$. From the undecidability of acceptance for queue 
  machines we can conclude  the undecidability of $\selsubtype$.
\qed
\end{proof}

%

%% file: impact.tex
Starting from our core undecidability result
we prove the undecidability of other subtyping relations
starting from $\subtype$.
In the
following subsections, for the sake of simplicity, we denote different
session type languages with the same letters $T,S, \dots$ and the
actual language will be made clear by the context.
%

\subsection{Undecidability of Asynchronous Subtyping}
The undecidability of the asynchronous \selection relation can be
exploited to show that asynchronous subtyping is undecidable.
Intuitively, this follows by the fact that the encoding of a queue
machine into $\selsubtype$ is also a valid encoding into $\subtype$.
\begin{corollary}
  Asynchronous subtyping for binary session types $\subtype$ is
  undecidable.
\end{corollary}
\begin{proof}
  A direct consequence of the undecidability of $\selsubtype$,
  i.e. $\subtype$ restricted to the pairs of types belonging to
  $(T^{\textsf{out}} \times T^{\textsf{in}}) \cap (\noinf \times \noinf)$.
%
\qed
\end{proof}

\subsection{Standard Binary Session Types with Asynchronous Subtyping}
The syntax proposed in Definition~\ref{def:sessiontypes} allows for
the definition of session types with only output selections and input branchings, and recursion. Standard binary session types~\cite{HVK98}
also feature classic send/receive types, containing the type of
the communicated message, dubbed {\em carried type}. Carried types can
be primitive types such as $\mathsf{bool}$, $\mathsf{nat}$, $\ldots$
or a session type $T$ (modelling delegation).  We define standard
binary session types as follows:
\begin{definition}[Standard Session Types]\label{pluto}
  Standard binary session types are defined as
  \begin{displaymath}
    \begin{array}{lrl}
      T &::=&   \ldots\text{as in Definition~\ref{def:sessiontypes}}\ldots\ \mid\
                \Tout U. T\
                \mid\ \Tin U. T
      \\
      U & ::= & T\mid \mathsf{bool}\mid\mathsf{nat}\mid\ldots 
  \end{array}
\end{displaymath}
\end{definition}

In order to extend subtyping to standard binary session types, we need
to adapt the notion of unfolding. The $n$-unfolding function is
extended as follows:
  \begin{displaymath}
    \begin{array}{l@{\qquad}l}
      \ldots\text{as in Definition~\ref{def:unfolding}}\ldots
      \\
      \unfold 1{\Tout U. T} = \; \Tout U. \unfold 1{T}
      \\
      \unfold 1{\Tin U. T } = \; \Tin U. \unfold 1{T }
    \end{array}
  \end{displaymath}
Moreover, the input context definition becomes:
  \[
  \mathcal A\ ::=\ \ldots\text{as in Definition~\ref{def:context}}\ldots\
  \quad\mid\quad \Tin U.\mathcal A
  \]

  Finally, the asynchronous subtyping relation for standard binary
  session types is given by the following:

  \begin{definition}\label{gigiociao}
    {\bf (Asynchronous Subtyping for Standard Session Types,} {\bf
      $\subtype_s$)} Asynchronous subtyping for standard session
    types is defined as in Definition~\ref{def:subtyping} where we 
    consider a relation $\mathcal R$ on both session and primitive 
    types that satisfies, besides the items in that definition, 
    also the following ones 
    (with an abuse of notation, we use $T$ and $S$ to range over
    both session and primitive types):
        
      \begin{itemize}
  

  \item[5.] if $T\in\{\mathsf{bool}\mid\mathsf{nat}\mid\ldots \}$ then
    $S=T$;


  \item[6.] if $T=\; \Tout U. T'$ then $\exists n\geq 0,\mathcal A$ such that
    \begin{itemize}
    \item
      $\unfold nS = \context A{\Tout {V_k}. S_k}{k\in \{1,\ldots,
        m\}}$,
    \item $\forall k\in\{1,\ldots, m\}. (V_k,U)\in\mathcal R$ and
    \item $(T,\context A {S_k}{k\in\{1,\ldots, m\}})\in\mathcal R$;
    \end{itemize}

  \item[7.] if $T=\;\Tin U. T'$ then $\exists n\geq 0$ such that
    $\unfold nS = \;\Tin V.S'$ and $(U,V)\in\mathcal R$ and
    $(T',S')\in\mathcal R$.




  \end{itemize}
\end{definition}

The undecidability proof applies to this extended
setting:
\begin{corollary}
  Asynchronous subtyping $\subtype_s$ for standard binary session types is
  undecidable.
\end{corollary}
\proof The core session type language in Definition~\ref{def:sessiontypes}
on which $\subtype$ is defined
is a fragment of the standard session type language
defined in the present section.  The thesis follows from the fact that
on this fragment the two subtyping relations $\subtype$ and
$\subtype_s$ coincide.  \qed

\paragraph{Remark}
It is worth to observe that the asynchronous subtyping relation
$\subtype_s$ defined above corresponds to the one by Mostrous and
Yoshida~\cite{IandCmostrous}, with the only difference that the
various rules for the different kinds of carried types considered in
that paper are replaced by the simpler rule {\em 5.}, that considers
primitive types only (carried session types are managed by the rules
{\em 6.} and {\em 7.}).  Hence our undecidability result applies also
to the subtyping by Mostrous and
Yoshida~\cite{IandCmostrous}.

\subsection{Carried Types in Selection/Branching and no Orphan
  Messages}
Chen et al.~\cite{MariangiolaPreciness} propose a variant of standard
binary session types where messages of some type can also be
communicated together with a choice performed within an output selection/input branching. 
This is defined by the following syntax.
\begin{definition}[Session types with carried types on choices
  \cite{MariangiolaPreciness}]\label{pippo}
  \begin{displaymath}
    \begin{array}{lrl}
      T &::=&   
              \TselectM{l}{T}{U}
              \mid \TbranchM{l}{T}{U} \mid
              \Trec t.T
              \mid \Tvar t
              \mid \Tend
      \\
      U & ::= & T\mid \mathsf{bool}\mid\mathsf{nat}\mid\ldots 
  \end{array}
\end{displaymath}
\end{definition}
This syntax corresponds with the one we have considered
in Definition \ref{def:sessiontypes}, where $\Tselect{l}{T}$ and $\Tbranch{l}{T}$
are replaced by $\TselectM{l}{T}{U}$ and $\TbranchM{l}{T}{U}$, respectively.
The definition of $n$-unfolding can be updated accordingly.

Chen et al.~\cite{MariangiolaPreciness} also propose a different
definition of subtyping that does not allow to have orphan messages,
i.e., inputs on the right-hand side of the subtyping relation cannot
be indefinitely delayed. We reformulate the subtyping
by Chen et al.~\cite{MariangiolaPreciness} (defined on infinite trees) as follows.
\begin{definition}[Asynchronous Orphan-Message-Free Subtyping, $\subtype_o$]
\label{def:relOrphFree}
Asynchronous subtyping for orphan-message-free session
    types is defined as in Definition~\ref{def:subtyping} where we 
    consider a relation $\mathcal R$ on both session and primitive 
    types that satisfies the items in that definition, with 
    items 2. and 3. replaced by the following
    corresponding ones, plus the additional rule 5.
    (also in this case, we use $T$ and $S$ to range over
    both session and primitive types):
  \begin{itemize}
  \item[2.] if $T=\TselectM{l}{T}{U}$ then
    $\exists n\geq 0,\mathcal A$ such that
    \begin{itemize}
    \item $\exists j\in I$ s.t. $T_j$ does not contain input branchings implies $\mathcal A=[\,]^1$,
    \item
      $\unfold nS = \context {A} {\TselectindexM
        l{S_k}{j}{J_k}{{V_k}}} {k\in \{1,\ldots, m\}}$,
    \item $\forall k\in\{1,\ldots, m\}.  I\subseteq J_k$ and
    \item
      $\forall i\in I,k\in\{1,\ldots, m\}. ({V_k}_i,U_{i})\in\mathcal
      R$;
    \item
      $\forall i\in I. (T_i,\context A {S_{ki}}{k\in\{1,\ldots, m\}})\in\mathcal R$;
    \end{itemize}

  \item[3.] if $T=\TbranchM{l}{T}{U}$ then $\exists n\geq 0$ such that\\
    $\unfold nS =$ \mbox{$\TbranchindexM lSjJV$}, $J\subseteq I$ and
    $\forall j\in J. (U_j,V_j),(T_j,S_j)\in\mathcal R$;
    
  \item[5.] if $T\in\{\mathsf{bool}\mid\mathsf{nat}\mid\ldots \}$ then
    $S=T$.
  \end{itemize}
\end{definition}

The key point in the definition above is rule {\em 2.}, first item,
that guarantees that if the r.h.s. does not start with the output 
needed to be mimicked in the simulation game, and then such
output must be anticipated, then all possible continuations in the
l.h.s. must contain at list an input. This implies that the
input in the r.h.s. that have been delayed due to the
anticipation, will be eventually involved in the simulation game,
i.e. they will be not delayed indefinitely.
It is worth to notice that this is guaranteed already by the
core relation $\selsubtype$ because it avoids processes
from having consecutive infinite outputs.

Also in this case the undecidability proof applies to this extended
setting:
\begin{corollary}
  Asynchronous subtyping $\subtype_o$ for binary session types with
  carried types in output selections/input branchings and asynchronous
  orphan-message-free subtyping is undecidable.
\end{corollary}
\proof First of all, we observe that our core session language has a
one-to-one correspondence with a fragment of the language given in
Definition~\ref{pippo}, under the assumption that only one given
primitive type can be carried (e.g. $\mathsf{bool}$).
Then, we observe (as already remarked above) that if we consider
terms without consecutive infinite outputs 
the new  additional first item 
of rule {\em 2.} in Definition \ref{def:relOrphFree} 
can be omitted without changing the defined relation.
In fact, we show that it is implied by the other conditions.
If during the simulation game
an output is anticipated w.r.t. some inputs in the r.h.s. 
(i.e. we use $A\neq[\,]^1$ in the application of rule {\em 2.})
then the continuations
of the simulation game could be either finite or infinite. For the finite
continuations we have that the inputs in front of the r.h.s. must be eventually 
consumed otherwise the pair $(\Tend,\Tend)$ cannot be reached; hence at least 
one input should be present in the l.h.s.
In the infinite continuations, the fact that the l.h.s. has no consecutive infinite outputs
guarantees the presence in such term of at least one input.

Hence we can conclude that $\selsubtype$, which is defined on terms
belonging to $\noinf$, is isomorphic to $\subtype_o$ restricted to 
terms without consecutive infinite output.
The undecidability of $\subtype_o$ thus directly follows from the undecidability of $\selsubtype$.
%
%
%
%
%
%
%
%
\qed

\paragraph{Remark} It is immediate to conclude that also the
subtyping relation by Chen et al.~\cite{MariangiolaPreciness} is
undecidable.

\subsection{Multiparty session types} \label{sec:multiparty}
We now investigate how our undecidability result can be applied to a
version of multiparty session types given by Mostrous et
al.~\cite{ESOP09}. Multiparty session types are an extension of binary
session types that allow to describe protocols between several
parties. Protocols, specified as {\em global types}~\cite{HYC16}, can
then be projected into {\em local types}, formally defined as follows.
\begin{definition}[Local Types]
  \begin{displaymath}
    \begin{array}{lrl}
      T &::=&   
              \ToutK U. T
              \mid 
              \TinK U. T
              \mid
              \TselectK{l}{T}
              \mid
              \TbranchK{l}{T}
              \mid 
              \Trec t.T
              \mid 
              \Tvar t
              \mid
              \Tend
      \\
      U & ::= & T\mid \mathsf{bool}\mid\mathsf{nat}\mid\ldots 
  \end{array}
\end{displaymath}
\end{definition}
Local types are a generalisation of the standard binary session types
seen in Definition~\ref{pluto}, where communications can now be
performed on different channels, e.g., a process involved in a session
with type $\ToutK U. k'?({U'}). T $ first outputs something of type
$U$ on channel $k$, and, then, inputs something of type $U'$ from
channel $k'$.

Before introducing our definition of subtyping for local types, we
note that the definition of $n$-unfolding can be trivially adapted to
terms in the definition above from our initial definition. Moreover,
we need to redefine input contexts for local types: such contexts 
now contain also outputs, under the assumption that those outputs
are on different channels. This reflects the fact that ordering is guaranteed
to be preserved only by messages sent on the same channel.
Technically, we use a parameterized notion of input context $\mathcal A^k$
where $k$ is assumed to be the channel of the output to be anticipated.

\begin{definition}[Multiparty input context]
  A multiparty input context $\mathcal A^k$ is a session type with
  multiple holes defined by the following syntax:
  \[
  \begin{array}{lll}
    \!\!\mathcal A^k\! ::= [\,]^n \ \mid\ k'\Tin U.\mathcal A^k \ \mid\
                       k'\Tbranch{l}{\mathcal A^k}\ \mid \ k''\Tout
                       U.\mathcal A^k\ \mid\ k''\Tselect{l}{\mathcal A^k}
  \end{array}
  \]
  where we assume that $k'$ is any possible channel while 
  $k'' \neq k$.
\end{definition}

We are now ready to define the subtyping relation for local types. The
definition of subtyping we propose is inspired by the one initially
proposed by Mostrous et al.~\cite{ESOP09}.  Unlike the binary case,
outputs on a channel can be anticipated over inputs and outputs on
different channels. Moreover, unlike Mostrous et al.~\cite{ESOP09}, we
allow outputs to be anticipated over inputs on the same channel, and
we do not allow inputs over different channels to be swapped. The
former point carries the same intuition as the output anticipation for
the binary case. The latter is a restriction that guarantees that
subtyping preserves the ordering of observable events (input actions)
given by the corresponding global type specification of the protocol.
The definition of asynchronous subtyping is then given as follows:

\begin{definition}[Multiparty Asynchronous Subtyping, $\subtype_m$]
\mbox{}\\
Asynchronous subtyping for multiparty session
    types is defined as in Definition~\ref{def:subtyping} where we 
    consider a relation $\mathcal R$ on both session and primitive 
    types that satisfies the items in that definition, with 
    items 2. and 3. replaced by the following
    corresponding ones, plus the additional items 5.--7.
    (also in this case, we use $T$ and $S$ to range over
    both session and primitive types):

  Asyncronous subtyping for multiparty session types is defined as in
  Definition~\ref{gigiociao}, by replacing the items about selection,
  branching, output and input with:

      \begin{itemize}
  
  \item[2.]
  
    if $T=\TselectK{l}{T}$ then $\exists n\geq 0,\mathcal A^k$ such
    that
    \begin{itemize}
    \item
      $\unfold nS = \context {\text{$\mathcal{A}^k$}} {k\Tselectindex l{S_h}{j}{J_h}} {h\in
        \{1,\ldots, m\}}$,
    \item $\forall h\in\{1,\ldots, m\}.  I\subseteq J_h$ and
    \item
      $\forall i\in I. (T_i,\context {A{\mbox{$^k$}}}{S_{hi}}{h\in\{1,\ldots, m\}})\in\mathcal R$;
    \end{itemize}

  \item[3.]
  
    if $T=\TbranchK{l}{T}$ then $\exists n\geq 0$ such that
    $\unfold nS = k\Tbranchindex lSjJ$, $J\subseteq I$ and
    $\forall j\in J. (T_j,S_j)\in\mathcal R$;

  \item[5.] if $T\in\{\mathsf{bool}\mid\mathsf{nat}\mid\ldots \}$ then
    $S=T$;

  \item[6.] if $T=\ToutK U. T'$ then $\exists n\geq 0,\mathcal A^k$
    such that
    \begin{itemize}
    \item $\unfold nS = \context {\text{$\mathcal{A}^k$}}{\ToutK
        {V_h}. S_h}{h\in \{1,\ldots, m\}}$,
    \item $\forall h\in\{1,\ldots, m\}. (V_h,U)\in\mathcal R$ and
    \item $ (T,\context {A{\mbox{$^k$}}}{S_h}{h\in\{1,\ldots, m\}})\in\mathcal R$;
    \end{itemize}
        
  \item[7.] if $T=\TinK U. T'$ then $\exists n\geq 0$ such that
    $\unfold nS = \TinK V.S'$ and $(U,V)\in\mathcal R$ and
    $(T',S')\in\mathcal R$.
        
  \end{itemize}

\end{definition}


As for the other cases, subtyping is undecidable:
\begin{corollary}
  Asynchronous subtyping for multiparty session types is undecidable.
\end{corollary}
\proof Similarly to the case for session types with carried types in
branching/selection, standard session types defined in
Definition~\ref{pluto} have a one-to-one correspondence with
a fragment of local types where only one single channel, e.g., $k$, is
used. Then, the thesis follows from the fact that on such fragment the
two subtyping relations $\subtype_s$ and $\subtype_m$ are isomorphic.
\qed

\paragraph{Remark} The proof of the previous Lemma does not directly
work for the definition of asynchronous subtyping used by Mostrous et
al.~\cite{ESOP09}. This is because their subtyping relation does not
allow to anticipate outputs over inputs on the same channel, e.g.,
$\ToutK U. k?({U'}). T $ is not a subtype of
$ k?({U'}). \ToutK U. T $. However, output anticipation is possible
over inputs on different channels, e.g., $\ToutK U. k'?({U'}). T$ is a
subtype of $k'?({U'}). \ToutK U. T$, assuming $k \neq k'$.  The
subtyping algorithm proposed by Mostrous et al.~\cite{ESOP09}
correctly checks cases like the above two examples, but fails to
terminate when there is an unbounded accumulation of inputs as in the
first example that we have discussed in \S\ref{sec:examples}.
Rephrasing that example in the syntax of local types, we have that
$T=\Trec t.k'?({U'}).\ToutK U.\Tvar t$ and
$S=\Trec t.k'?({U'}).k'?({U'}).\ToutK U.\Tvar t$ are in subtyping
relation, even for Mostrous et al.~\cite{ESOP09}.  Nevertheless, the
algorithm proposed in that paper does not terminate because, in this
case, it is expected to check infinitely many different pairs
$(T,k'?({U'}).S)$, $(T,k'?({U'}).k'?({U'}).S)$, $\dots$.
%
In the light of our undecidability result, we can even conclude the
impossibility to check algorithmically the subtyping relation by
Mostrous et al.~\cite{ESOP09}.  Consider the proof of our
Theorem~\ref{thm}: given a queue machine, we can change both encodings
of its finite control and its queue so that all inputs are on some
special channel $k$ and all outputs are on some special channel $k'$,
with $k\neq k'$. As discussed above, outputs on $k'$ can be
anticipated w.r.t. inputs on a different channel $k$, hence the two
encodings will be in a subtyping relation, also for the subtyping
by Mostrous et al.~\cite{ESOP09}, if and only if the encoded machine does not
terminate.

\subsection{Communicating automata}

A Communicating Finite State Machine (CFSM) \cite{BZ83}, or more simply a communicating automaton, is defined as a finite automaton $(Q,q_0,\Sigma, \delta)$, where 
\begin{itemize}
\item $Q$ is a finite set of states 
\item $q_0 \in Q$ is the initial state
\item $\Sigma$ is a finite alphabet, and
\item $\delta \subseteq Q \times \Sigma \times \{!,?\} \times Q$ is a transition set.
\end{itemize}
We use ``$?$'' to represent inputs and ``$!$'' to represent outputs
(in CFSMs \cite{BZ83} ``$+$'' and ``$-$'', respectively, are used, instead).

A CFSM is a 
labeled transition system 
that can be employed to graphically represent a session type (see, e.g., Figures \ref{fig:queue} and \ref{fig:control}). 
Note that in general a CFSM may express more
behaviours than the ones described by session types: 
it can include non-deterministic and mixed choices, i.e.\ choices including both inputs and outputs.

Let $\calt$ be the set of all session types $T$ and $L$ the alphabet
of session types, we define a transition relation
$\longrightarrow \; \subseteq \calt \times L \times \{!,?\} \times
\calt$, as the least transition set satisfying the following rules 
$$\begin{array}{c}
\Tselect{l}{T} \arrow{l_i !} T_i \hspace{.5cm} i \in I 
\\[.2cm]
 \Tbranch{l}{T} \arrow{l_i ?} T_i \hspace{.5cm} i \in I 
\\[.2cm]
\infr{ T\{\Trec t.{T}/\Tvar{t}\} \arrow{\alpha} T'}{\Trec t.{T} \arrow{\alpha} T' } 
\end{array}
$$
with $\alpha$ ranging over $ L \times \{!,?\}$.

Given a session type $T$ we define $CFSM(T)$ as being the communicating automaton $(Q_T,T,L, \delta_T)$, where:
$L$ is the alphabet of session types, $Q_T$ is the set of terms $T'$ which are reachable from $T$ according to $\longrightarrow$ relation and $\delta_T$ is defined as the restriction of $\longrightarrow$ to  $Q_T \times L \times \{!,?\} \times Q_T$. For example Figure \ref{fig:queue} depicts $CFSM(S)$ with $S$ being the session type defined in
Definition \ref{def:queueenc} (assuming $\Gamma = \{ A_i \mid i \leq k \}$).

We, thus, get the following result as a consequence of undecidability of $\subtype$. Any relation $\preceq$ 
over communicating automata (usually called {\it refinement} relation in this context) that is such that 
$CFSM(T) \preceq CFSM(T')$ if and only if $T \, \subtype \, T'$, i.e.\ it reduces to our subtyping definition for the subclass of communicating automata not including non-deterministic and mixed choices, is undecidable.

%

%% file: decidable.tex
We now show that we cannot further reduce (w.r.t.\ branching/selection structure) the core
undecidable fragment: if we consider single-output selection only or
single-input branching only, we obtain a decidable relation.

\begin{definition}[Asynchronous \Selection Output Relation]\label{def:outrel} \mbox{}\\
  The asynchronous \selection output 
  relation
  $\selsubtype_{\mathsf{sout}}$ is defined as:
  \[ \selsubtype_{\mathsf{sout}}\ =\  \selsubtype\cap (T^{\textsf{out}} \times T^{\textsf{out}})\]
\end{definition}

\begin{definition}[Asynchronous \Selection Input Relation]\label{def:inrel} \mbox{}\\
  The asynchronous \selection input 
  relation
  $\selsubtype_{\mathsf{sin}}$ is defined as:
  \[ \selsubtype_{\mathsf{sin}}\ =\  \selsubtype\cap (T^{\textsf{in}} \times T^{\textsf{in}})\]
\end{definition}

  As a matter of fact, we prove decidability for larger relations w.r.t. $\selsubtype_\mathsf{sin}$
  and $\selsubtype_\mathsf{sout}$ where we do not impose the constraint about
  no consecutive infinite outputs.

In order to define an algorithm for deciding the two relations above,
we first adapt to our setting the procedure defined in Mostrous et. al~\cite{ESOP09} 
and then improve it to precisely characterize the two relations.
The initial procedure is defined for the unrestricted syntax of session
types, i.e. the subtyping $\subtype$; while the improved version assumes 
to work on types restricted according to the single-choice assumptions, 
i.e. the two new relations $\selsubtype_{\mathsf{sin}}$ and $\selsubtype_{\mathsf{sout}}$.
Actually, in order to have a more general decidability result,
 we show that it is not necessary to consider the constraint about
  no consecutive infinite outputs.

The procedure is defined by the rules reported in
Figure~\ref{fig:algo}.
\begin{figure}[t]
  \begin{displaymath}
    \begin{array}{cccc}
      \infer[\textsf{Asmp}]
      {
      \Sigma, (T,S) \vdash T\subtypea S
      }
      {
      }
      &
        \infer[\textsf{End}]
        {
        \Sigma \vdash \Tend\subtypea\Tend
        }
        {
        }
      \\\\
      \infer[\textsf{Out}]
      {
      \Sigma \vdash \Tselect{l}{T} \subtypea \context A{\Tselectindex{l}{S_n}j{J_n}}{n}
      }
      {
      \forall n.I\subseteq J_n & \forall i\in I\,.\,\Sigma \vdash T_i\subtypea \context A {S_{ni}}{n} 
          }
        &
          \infer[\textsf{In}]
          {
          \Sigma \vdash \Tbranch{l}{T} \subtypea \Tbranchindex{l}{S}jJ
          }
          {
          J\subseteq I & \forall j\in J\,.\,\Sigma \vdash T_j\subtypea S_j
                         }
      \\\\
      \multicolumn{2}{c}{
      \infer[\textsf{RecL}]
      {
      \Sigma \vdash \Trec t.T \subtypea S
      }
      {
      \Sigma,(\Trec t.T,S) \vdash \unfold 1{\Trec t.T} \subtypea S
      }
      }
      \\\\
      \multicolumn{2}{c}{
      \infer[\textsf{RecR}_1]
      {
      \Sigma \vdash T\subtypea \Trec t.S
      }
      {
      T=\Tend \vee T=\Tbranch lT & \Sigma, (T,\Trec t.S) \vdash T\subtypea \unfold 1{\Trec t.S}
                                   }
                                   }
      \\\\
      \multicolumn{2}{c}{
      \infer[\textsf{RecR}_2]
      {
      \Sigma \vdash \Tselect lT\subtypea S
      }
      {
      n=\mathsf{depth}(S,\emptyset)
      & n \geq 1  &
                    \Sigma, (\Tselect lT,S) \vdash \Tselect lT\subtypea \unfold nS
                    }
                    }
    \end{array}
  \end{displaymath}
  \caption{A Procedure for Checking Subtyping}
  \label{fig:algo} 
\end{figure}
In the rules, the environment $\Sigma$ is a set of pairs $(T,S)$ used
to keep track of previously visited pairs of types. The procedure
successfully terminates either by applying rule \textsf{Asmp} (which
has priority over the other rules) or by applying rule
\textsf{End}. In the former case, a previously visited pair is being
visited again, and therefore, there is no need to proceed further.
Rules \textsf{In} and \textsf{Out} are straightforward. The procedure
always unfolds on the left-hand side when necessary (rule
\textsf{RecL}). If this is not the case, but it is necessary to unfold
on the right-hand side, it is possible to apply either
$\textsf{RecR}_1$ or $\textsf{RecR}_2$, depending on whether the
left-hand side type is an input (or an $\Tend$) or an output,
respectively. In the first case, a single unfolding is
sufficient. However, we may need to unfold several times in the case
we need an output. The partial function $\depth$
(we write $\depth(S)= \; \perp$ if $\depth$ is undefined on $S$)
 measures the number of
unfoldings necessary for anticipating such output. The function is
inductively defined as:
\begin{displaymath}
  \begin{array}{cc}
    \depth(\Tend,\Gamma) = \; \perp \qquad \qquad \depth(\Tselect
    lT,\Gamma) = 0\\
    \depth(\Tbranch lT,\Gamma) = 
    \mathsf{max}\{\depth(T_i,\Gamma)\ |\ i\in I\} 
    \\
    \depth(\Trec t.T,\Gamma) = \left\{
      \begin{array}{ll}
        \perp & 
                  \text{if } \Tvar t\in\Gamma 
        \\
        1+\depth(T\{\Trec t.T/\Tvar t\}, \Gamma+\{\Tvar t\}) & \text{otherwise}
      \end{array}
      \right.
\end{array}
\end{displaymath}
In the definition of $\depth$, we assume that
$\mathsf{max}\{\depth(T_i,\Gamma)\ |\ i\in I\}=\;\perp$ if
$\depth(T_i,\Gamma)=\;\perp$ for some $i\in I$. Similarly,
$1 \,+\perp \;= \;\perp$.

The subtyping procedure, when it has to check whether $T \leq S$,
applies the rules from bottom to top starting from the judgement
$\emptyset \vdash T \subtypea S$.  We write
$\Sigma \vdash T \subtypea S \rderiv \Sigma' \vdash T' \subtypea S'$
if $\Sigma \vdash T \subtypea S$ matches the consequences of one of
the rules, and $\Sigma' \vdash T' \subtypea S'$ is produced by the
corresponding premises.
$\Sigma \vdash T \subtypea S \rderiv^* \Sigma' \vdash T' \subtypea S'$
is the reflexive and transitive closure of such relation.  We write
$\Sigma \vdash T \subtypea S \notrderiv$ to mean that no rule can be
applied to the judgement $\Sigma \vdash T \subtypea S$.  We give priority to the application of
the rule $\textsf{Asmp}$ to have a deterministic procedure: this is
sufficient because all the other rules are alternative, i.e., given a
judgement $\Sigma \vdash T \subtypea S$ there are no two rules that
can be both applied.

We now prove that the above procedure
is a semi-algorithm for checking whether two types are not 
in subtyping relation.

\begin{lemma}\label{lem:semidecidable}
  Given the types $T$ and $S$ we have
  $\exists \Sigma', T', S' \ldotp \; \emptyset \vdash T \subtypea S
  \rderiv^* \Sigma' \vdash T' \subtypea S' \notrderiv$
  if and only if $T \not\!\!\subtype \, S$.
\end{lemma} 
\proof We prove the two implications separately.  We start with the
{\em if} part and proceed by contraposition.
Assume that it is not true that
\mbox{$\exists \Sigma'\!, T'\!, S' \ldotp \, \emptyset \vdash T \subtypea S$}
$\rderiv^* \Sigma' \vdash T' \subtypea S' \notrderiv$.
Consider now the relation
$${\mathcal R}=\{(T',S')\ |\ \exists \Sigma' \ldotp \emptyset \vdash T
\subtypea S \rderiv^* \Sigma' \vdash T' \subtypea S'\}$$
We show that ${\mathcal R}$ is a subtyping relation. Let
$(T',S') \in {\mathcal R}$.
Then, it is possible to apply at least one rule to
$\Sigma' \vdash T' \subtypea S'$ for the environment $\Sigma'$ such
that
$\emptyset \vdash T \subtypea S \rderiv^* \Sigma' \vdash T' \subtypea
S'$.  We proceed by cases on $T'$. In the following we use $\mathsf{nrec}(R)$
to denote the number of unguarded --not prefixed
    by some branching/selection-- occurrences of recursions $\Trec t. R'$ in
    $R$ for any $R',\Tvar t$.
  \begin{itemize}

  \item \label{item:end} If $T'=\Tend$ then item $1$ of Definition
    \ref{def:subtyping} for pair $(T',S')$ is shown by induction on
    $k=\mathsf{nrec}(S')$.
  \begin{itemize}
  \item Base case $k=0$. The only rule applicable to
    $\Sigma' \vdash T' \subtypea S'$ is $\textsf{End}$, that immediately
    yields the desired pair of ${\mathcal R}$.
  \item Induction case $k>0$. The only rules applicable to
    $\Sigma' \vdash T' \subtypea S'$ are $\textsf{Asmp}$ and
    $\textsf{RecR}_1$.  In the case of $\textsf{Asmp}$ we have that
    $(T',S') \in \Sigma'$, hence there exists $\Sigma''$ with
    $(T',S') \notin \Sigma''$ such that
    $\emptyset \vdash T \subtypea S \rderiv^* \Sigma'' \vdash T'
    \subtypea S' \rderiv^* \Sigma' \vdash T' \subtypea S'$
    and rule $\textsf{RecR}_1$ has been applied to
    $\Sigma'' \vdash T' \subtypea S'$.
    So for some $\Sigma'''$ ($=\Sigma'$ or $=\Sigma''$) we have that
    the procedure applies rule $\textsf{RecR}_1$ to
    $\Sigma''' \vdash T' \subtypea S'$. Hence
    \mbox{$\Sigma''' \vdash T' \subtypea S' \rderiv \Sigma'''' \vdash T'
    \subtypea \unfold 1{S'} $}
    . Since $\mathsf{nrec}(\unfold 1{S'})=k-1$, by induction
    hypothesis item $3$ of Definition \ref{def:subtyping} holds for
    pair $(T', \unfold 1{S'})$, hence it holds for pair $(T',S')$.
  \end{itemize}

\item \label{item:internal} If $T'=\Tselect{l}{T}$ then item $2$ of
  Definition \ref{def:subtyping} for pair $(T',S')$ is shown as
  follows:
\begin{itemize}
\item If $\mathsf{depth}(S,\emptyset)=0$ then the only rule applicable
  to $\Sigma' \vdash T' \subtypea S'$ is $\textsf{Out}$, that
  immediately yields the desired pairs of ${\mathcal R}$.
\item If $\mathsf{depth}(S,\emptyset)\geq1$ then the only rules
  applicable to $\Sigma' \vdash T' \subtypea S'$ are $\textsf{Asmp}$
  and $\textsf{RecR}_2$.  In the case of $\textsf{Asmp}$ we have that
  $(T',S') \in \Sigma'$, hence there exists $\Sigma''$ with
  $(T',S') \notin \Sigma''$ such that
  $\emptyset \vdash T \subtypea S \rderiv^* \Sigma'' \vdash T'
  \subtypea S' \rderiv^* \Sigma' \vdash T' \subtypea S'$
  and rule $\textsf{RecR}_2$ has been applied to
  $\Sigma'' \vdash T' \subtypea S'$.  So for some $\Sigma'''$
  ($=\Sigma'$ or $=\Sigma''$) we have that the procedure applies rule
  $\textsf{RecR}_2$ to $\Sigma''' \vdash T' \subtypea S'$. Hence
  \mbox{$\Sigma''' \vdash T' \subtypea S' \rderiv \Sigma'''' \vdash T'
  \subtypea \unfold k{S'} $},
  taking $k = \mathsf{depth}(S,\emptyset)$.  Since
  $\mathsf{depth}(\unfold k{S'})=0$, we end up in the previous
  case. Therefore item $3$ of Definition \ref{def:subtyping} holds for
  pair $(T', \unfold k{S'})$, hence it holds for pair $(T',S')$.

\end{itemize}

\item \label{item:external} If $T'=\Tbranch{l}{T}$ then item $3$ of
  Definition \ref{def:subtyping} for pair $(T',S')$ is shown by
  induction on $k=\mathsf{nrec}(S')$:
  \begin{itemize}
  \item Base case $k=0$. The only rule applicable to
    $\Sigma' \vdash T' \subtypea S'$ is $\textsf{In}$, that immediately
    yields the desired pairs of ${\mathcal R}$.
  \item Induction case $k>0$. The only rules applicable to
    $\Sigma' \vdash T' \subtypea S'$ are $\textsf{Asmp}$ and
    $\textsf{RecR}_1$.  In the case of $\textsf{Asmp}$ we have that
    $(T',S') \in \Sigma'$, hence there exists $\Sigma''$ with
    $(T',S') \notin \Sigma''$ such that
    $\emptyset \vdash T \subtypea S \rderiv^* \Sigma'' \vdash T'
    \subtypea S' \rderiv^* \Sigma' \vdash T' \subtypea S'$
    and rule $\textsf{RecR}_1$ has been applied to
    $\Sigma'' \vdash T' \subtypea S'$.  So for some $\Sigma'''$
    ($=\Sigma'$ or $=\Sigma''$) we have that the procedure applies
    rule $\textsf{RecR}_1$ to $\Sigma''' \vdash T' \subtypea S'$.
    Hence
    $\Sigma''' \vdash T' \subtypea S' \rderiv \Sigma'''' \vdash T'
    \subtypea \unfold 1{S'} $
    . Since $\mathsf{nrec}(\unfold 1{S'})=k-1$, by induction
    hypothesis item $3$ of Definition \ref{def:subtyping} holds for
    pair $(T', \unfold 1{S'})$, hence it holds for pair $(T',S')$.

  \end{itemize}

\item \label{item:rec} If $T'= \Trec t.{T'} $ then item $4$ of
  Definition \ref{def:subtyping} for pair $(T',S')$ holds because the
  only rule applicable to $\Sigma' \vdash T' \subtypea S'$ is
  $\textsf{RecL}$ that immediately yields the desired pair of
  ${\mathcal R}$.

  \end{itemize}

  We now prove the {\em only if} part and proceed by contraposition.
  Assume that there exists a relation ${\mathcal R}$ that is a
  subtyping relation such that $(T,S) \in {\mathcal R}$.

  We say that
  $\Sigma \vdash T \subtypea S \rderiv_w \Sigma' \vdash T' \subtypea
  S'$
  if
  $\Sigma \vdash T \subtypea S \rderiv^* \Sigma' \vdash T' \subtypea
  S'$ and: the last rule applied is one of
  $\textsf{Out}$, $\textsf{In}$ or $\textsf{RecL}$ rules; while all
  previous ones are $\textsf{RecR}_1$ or $\textsf{RecR}_2$ rules.

  We start by showing that if
  $\exists \Sigma \ldotp \emptyset \vdash T \subtypea S \rderiv_w^*
  \Sigma \vdash T' \subtypea S'$
  implies $\exists n \ldotp (T',\unfold n{S'}) \in {\mathcal R}$. The
  proof is by induction on the length of such computation $\rderiv_w^*$
  of the procedure.  The base case is for a $0$ length computation: it
  yields $(T,S) \in {\mathcal R}$ which holds.  For the inductive case
  we assume it holds for all computations of a length $k$ and we
  show it holds for all computations of length $k+1$, by considering
  all judgements $\Sigma' \vdash T'' \subtypea S''$ such that
  $\Sigma \vdash T' \subtypea S' \rderiv_w \Sigma' \vdash T''
  \subtypea S''$.
  This is shown by first considering the case in which rule
  $\textsf{Asmp}$ applies to $\Sigma \vdash T' \subtypea S'$: in this
  case there is no such a judgement and there is nothing to prove.  Then
  we consider the case in which $T'=\Tend$ and
  $\Sigma \vdash \Tend \subtypea S' \rderiv^* \Sigma''' \vdash \Tend
  \subtypea \Tend$
  (by applying $\textsf{RecR}_1$ rules) and rule $\textsf{End}$
  applies to $\Sigma''' \vdash \Tend \subtypea \Tend$. Also in this
  case there is no such a judgement $\Sigma' \vdash T'' \subtypea S''$
  and there is nothing to prove.  Finally, we proceed by an immediate
  verification that judgements $\Sigma' \vdash T'' \subtypea S''$
  produced in remaining cases are required to be in $ {\mathcal R}$ by
  items $2$, $3$ and $4$ of Definition \ref{def:subtyping}:
  $T'=\Tselect{l}{T}$ ($\rderiv_w$ is a possibly empty sequence of
  $\textsf{RecR}_2$ applications followed by $\textsf{Out}$
  application), $T'=\Tbranch{l}{T}$ ($\rderiv_w$ is a possibly empty
  sequence of $\textsf{RecR}_1$ applications followed by $\textsf{In}$
  application) or $T'= \Trec t.{T'}$ ($\rderiv_w$ is simply
  $\textsf{RecL}$ application).

  We finally observe that, considered any judgement
  $\Sigma \vdash T' \subtypea S'$ such that
  $\exists n \ldotp (T',\unfold n{S'}) \in {\mathcal R}$, we have:
  \begin{itemize}
  \item either rule $\textsf{Asmp}$ applies to
    $\Sigma \vdash T' \subtypea S'$, or

  \item $T'=\Tend$ and, by item $1$ of Definition
    \ref{def:subtyping}, there exists $\Sigma'$ such that
    $\Sigma \vdash \Tend \subtypea S' \rderiv^* \Sigma' \vdash \Tend
    \subtypea \Tend$
    (by applying $\textsf{RecR}_1$ rules) and rule $\textsf{End}$ is
    the unique rule applicable to
    $\Sigma' \vdash \Tend \subtypea \Tend$, with $\textsf{RecR}_1$ being
    the unique rule applicable to intermediate judgements, or

  \item by items $2$,$3$ and $4$ of Definition
    \ref{def:subtyping}, there exist $\Sigma', T'',S''$ such that
    $\Sigma \vdash T' \subtypea S' \rderiv_w^* \Sigma' \vdash T''
    \subtypea S''$, with each intermediate judgement
    having a unique applicable rule.
    In particular this holds for $T'=\Tselect{l}{T}$ ($\rderiv_w$ is a
    possibly empty sequence of $\textsf{RecR}_2$ applications followed
    by $\textsf{Out}$ application), $T'=\Tbranch{l}{T}$ ($\rderiv_w$
    is a possibly empty sequence of $\textsf{RecR}_1$ applications
    followed by $\textsf{In}$ application) or $T'= \Trec t.{T'}$
    ($\rderiv_w$ is simply $\textsf{RecL}$ application). 
    \qed
  \end{itemize}

%



The above procedure is not guaranteed to terminate when $T \subtype S$,
in particular in those cases in which an infinite subtyping relation
is needed to prove that they are in subtyping. For instance, this
happens for the cases discussed in the examples reported in Section 
\ref{sec:examples}. 
We now show how to amend the procedure to terminate at least in the
restricted cases discussed above.  The new algorithm is defined by the
same rules as before, plus a pair of rules presented below, where we
use judgements $\Sigma \vdash T \subtypet S$ instead of
$\Sigma \vdash T \subtypea S$.  

Moreover, the new version of the
algorithm requires to distinguish among different instances of the
same input branching.  More precisely, due to multiple
unfoldings 
we could have that the same input
choice appears more than once.
For instance, given $\Trec t. \Tbranchsingle {l}{\Tvar t}$,
if we apply unfolding twice, we obtain 
$\Tbranchsingle {l}{\Tbranchsingle {l}{\Trec t. \Tbranchsingle {l}{\Tvar t}}}$.
To distinguish different instances of the same input branching, 
we assume to have an extended syntax in which such choices are annotated:
$\Tbranchsingledec {l}{S}{\alpha}$ denotes an input annotated with a symbol
$\alpha$ taken from a countable set of annotations. Annotations in the same
term are assumed to be pairwise distinct.
In the example above, the decorated version of the unfolded term is
$\Tbranchsingledec {l}{\Tbranchsingledec {l}{\Trec t. \Tbranchsingle {l}{\Tvar t}}{\alpha'}}{\alpha}$,
for a pair of distinct annotations $\alpha$ and $\alpha'$.
As multiple unfoldings can be applied only to the r.h.s. terms
of the judgements $\Sigma \vdash T \subtypet S$, we will adopt the syntax
extended with annotations only for such terms.

Algorithmically, in order to have the guarantee that all the annotations are pairwise
distinct, we assume that they are all different in the initial term.
Namely, when we want to check
whether $T \, \selsubtype_{\mathsf{sout}} \, S$ or $T \, \selsubtype_{\mathsf{sin}} \, S$
we assume an annotation function $\decore{S}$ that annotates all the
input choices in $S$ with pairwise distinct symbols, and 
we start the algorithm from the judgement
$\emptyset \vdash T \subtypet \decore{S}$.  Moreover, every time an unfolding
is applied to the r.h.s.  of a
judgement, the algorithm annotates with fresh symbols each added input
choice.

Concerning annotations, we omit them when they are irrelevant.
For instance, the rules in Figure \ref{fig:algo} do not 
contain annotations, but are used any way to define also the
new algorithm (upon replacement of $\Sigma \vdash T \subtypea S$
with $\Sigma \vdash T \subtypet S$). The omission of the annotations
means that the rules in that figure can be
applied to any annotated judgement $\Sigma \vdash T \subtypet S$,
to obtain a new judgement $\Sigma' \vdash T' \subtypet S'$
where the new annotated term $S'$ inherits the annotations of $S$
and in case new input branchings are present in $S'$
(due to unfolding in rules $\textsf{RecR}_1$
and $\textsf{RecR}_2$) these will
be freshly annotated as discussed above.

We assume a function $\undecore{S}$ that returns the same term but
without annotations. We overload the function $\undecore{\_}$ to apply
it also to an environment $\Sigma$: $\undecore{\Sigma}$ is the
environment obtained by removing all annotations from the r.h.s. of
all of its pairs.


We are now ready to present the two additional rules, having the same
higher priority of $\textsf{Asmp}$:
\begin{displaymath}
  \begin{array}{lll}
    \infer[\textsf{Asmp}2]
    {
    \begin{array}{c}
    \Sigma, (T,    \Tbranchsingle {l_1}
    {
    \ldots\Tbranchsingledec {l_r}
    {
    \ldots\Tbranchsingle {l_n}{S}\ldots
    }{\alpha}
    \ldots
    }) \vdash
    \\
    T\subtypet 
    \Tbranchsingledec {l_1}
    {
    \ldots\Tbranchsingle {l_m}
    {
    S'
    }
    \ldots
    }{\alpha}
    \end{array}
    }
    {
\begin{array}{c}    
j > i
\quad 
s < |\gamma|
\quad 
l_1\cdots l_n = \gamma^i\cdot(l_1\cdots l_{s})
\\
l_1\cdots l_m = \gamma^j\cdot(l_1\cdots l_{s})
    \quad
    S\in\{\oplus, \mu\}
    \quad
                     \& \in T 
\quad
  \undecore{S}=\undecore{S'}
  \end{array}
  }    
  \end{array}
\end{displaymath}

\begin{displaymath}
  \begin{array}{lll}
    \infer[\textsf{Asmp}3]
    {
    \Sigma, (T,    \Tbranchsingle {l_1}
    {
    \ldots\Tbranchsingle {l_n}
    {
    S
    }
    \ldots
    }) \vdash T\subtypet 
    \Tbranchsingle {l_1}
    {
    \ldots\Tbranchsingle {l_m}
    {
    S'
    }
    \ldots
    }
    }
    {
    S\in\{\oplus, \mu\}
    \qquad
    \& \not\in T
    \qquad
    n < m
    \qquad
    \undecore{S}=\undecore{S'}
    }
  \end{array}
\end{displaymath}

\medskip


Above $S\in\{\oplus, \mu\}$ means that $S$ starts
with either an output selection or a recursive definition,
while $\& \in T$ requires the occurrence of at least
one input branching anywhere in the term $T$.

Intuitively, $\textsf{Asmp}2$ allows the algorithm to terminate when,
after the judgement $\Sigma \vdash T \subtypet R$, another
judgement $\Sigma' \vdash T \subtypet R'$ is reached such that:
\begin{itemize}
\item
both $R$ and $R'$ start with a sequence of input branchings
(respectively labeled with $l_1 \cdots l_n$ and $l_1 \cdots l_m$);
\item
the label sequences are repetitions of the same pattern, with the 
second one strictly longer than the first one (namely,
there exist $\gamma$, a proper prefix $l_1 \cdots l_s$ 
of $\gamma$ and $j > i$ such that $l_1 \cdots l_n = \gamma^i \cdot
(l_1 \cdots l_s)$ and $l_1 \cdots l_m = \gamma^j \cdot
(l_1 \cdots l_s)$);

\item there exists $l_r$ in the initial input sequence of $R$ having an annotation $\alpha$ that
coincides with the annotation of $l_1$  in the initial input sequence of $R'$ (hence
$l_1=l_r$, $n \geq 1$ and $m \geq 1$)

\item
after such initial input branchings, both $R$ and $R'$ 
continue with the same term up to annotations ($S$ and $S'$ such that  $\undecore{S}=\undecore{S'}$).
\end{itemize}
The algorithm can terminate in this case because otherwise
it would continue indefinitely by repeating the same
steps performed between the judgements \mbox{$\Sigma \vdash T \subtypet R$} and 
$\Sigma' \vdash T \subtypet R'$. During these steps
the l.h.s.\ term $T$ performs a cycle including: 
the non-empty sequence of inputs $l_1 \cdots l_{r-1}$ (due to the requirement above about the $\alpha$ annotation and the constraint $\& \in T$ in the rule premise) that is a repetition of $\gamma$,
which is matched by the initial part of the input sequence in the r.h.s.;
and a sequence of outputs, which is matched by the final term $S$
of the r.h.s.\ that, itself, performs a cycle.
While performing this cycle, $S$ generates new input branchings that strictly extend the initial input sequence.

The rule $\textsf{Asmp}3$ allows the algorithm to terminate in the
case the potential infinite repetition includes output selections only. In this case it is sufficient to check that 
the l.h.s.\ cycles without consuming any inputs and that,
after a number of such cycles: the amount of  
initial accumulated inputs in the r.h.s.\ strictly increases
and the term after such accumulation performs, itself, a cycle.

Also for the new algorithm we use 
$\Sigma \vdash T \subtypet S \rderiv^* \Sigma' \vdash T' \subtypet S'$
to denote the application of one rule on the former
judgement that generates the latter,
and $\Sigma \vdash T \subtypet S \notrderiv$ to denote
that no rule can be applied to the judgement $\Sigma \vdash T \subtypet S$.

We now prove that the new algorithm terminates.

\begin{lemma}\label{lemma:termination}
Given two types 
$T \in T^{\textsf{out}} \, \cap \, T^{\textsf{in}}$ (resp. $T \in T^{\textsf{out}}$)
and 
\mbox{$S \in T^{\textsf{in}}$} (resp. $S \in T^{\textsf{in}} \, \cap \, T^{\textsf{out}}$), 
the algorithm applied to initial judgement $\emptyset \vdash T \subtypet \decore{S}$
terminates.
\end{lemma} 
\proof In this proof, we abstract away from the annotations of input
actions, i.e., we denote two types that differ only in the annotations
with the same term.

We proceed by contraposition.  Assume that there exist $T$ and $S$
such that the algorithm starting from the initial judgement
$\emptyset \vdash T \subtypet S$ does not terminate.  Hence, there
must exist an infinite sequence of rule applications
$\Sigma \vdash T \subtypet S \rderiv \Sigma_1 \vdash T_1 \subtypet S_1
\rderiv \ldots \rderiv \Sigma_n \vdash T_n \subtypet S_n \rderiv
\ldots$.
In this sequence, infinitely many unfoldings of recursive definitions
will be performed on infinitely many traversed judgements
$\Sigma' \vdash T' \subtypet S'$.  All these pairs $(T',S')$ must be
pairwise distinct otherwise rule $\textsf{Asmp}$ will be applied to
terminate the sequence.  This implies that the environment $\Sigma$
grows unboundedly in order to contain all these infinitely many
distinct pairs $(T',S')$.

We now prove that all such pairs $(T',S')$ are of the following form:
$(T_f,$ \linebreak \mbox{$\Tbranchsingle {l_1}{\ldots\Tbranchsingle {l_n}{S_f}\ldots})$}
where $T_f$ and $S_f$ belong to a finite set of terms.  The first term
$T'$ is obtained from the initial term $T$ by means of consumptions of
initial inputs or outputs (rules $\textsf{In}$ and $\textsf{Out}$), or
single unfoldings of top level recursive definitions (rules
$\textsf{RecL}$).  The set of terms that can be obtained in this way
from the finite initial term $T$ is clearly finite.  On the contrary,
the second group of terms $S'$ could be obtained by application of
different transformations: in particular $\textsf{Out}$ that allows
for the anticipation of outputs prefixed by an arbitrary sequence of
inputs, and $\textsf{RecR}_2$ that can unfold more than one recursive
definition at a time.  Concerning the multiple unfoldings of
$\textsf{RecR}_2$, we have that the $\depth(\_)$ function guarantees
that recursive definitions guarded by an output operation are never
unfolded. In this way a subterm of $S'$ prefixed by an output is taken
from a finite set of terms.  Obviously, also the set of possible
subterms starting with a recursive definition is finite as well.  As
$S'$ has only input single-choices, we can conclude that it is of the
form $\Tbranchsingle {l_1}{\ldots\Tbranchsingle {l_n}{S_f}\ldots}$
because it starts with a (possibly empty) sequence of inputs followed
by a term $S_f$, guarded by an output or a recursive definition, taken
from a finite set.

As infinitely many distinct pairs $(T',S')$ are introduced in
$\Sigma$, there are infinitely many distinct pairs having the same
$T_f$ and $S_f$.  This guarantees the existence of an infinite
sequence
$$
\begin{array}{l}
  (T_f,\Tbranchsingle {l^1_1}{\ldots\Tbranchsingle {l^1_{n_1}}{S_f}\ldots})\\
  (T_f,\Tbranchsingle {l^2_1}{\ldots\Tbranchsingle {l^2_{n_2}}{S_f}\ldots})\\ \cdots\\
  (T_f,\Tbranchsingle {l^v_1}{\ldots\Tbranchsingle {l^v_{n_v}}{S_f}\ldots})\\ \cdots
\end{array}
$$
of pairs that are introduced in $\Sigma$ in this order, and for which
$n_v$ is strictly growing.

If $T_f$ does not contain input actions (i.e. $\& \not\in T_f$) we
have that the rule $\textsf{Asmp}3$ could be applied when the second
pair in the above sequence 
$(T_f,$
\linebreak
\mbox{$\Tbranchsingle {l^2_1}{\ldots\Tbranchsingle
  {l^2_{n_2}}{S_f}\ldots})$}
was introduced in $\Sigma$. As $\textsf{Asmp}3$ has priority, it is
necessary to apply this rule, thus terminating successfully the
sequence of rule applications. This contradicts the initial assumption
about the infinite sequence of rule applications.
 
Now, consider $T_f$ that contains at least one input action.  We
separate this case in two subcases:
$T_f \in T^{\textsf{out}} \cap T^{\textsf{in}}$ and
$T_f \in T^{\textsf{out}}$.

In the first subcase, both inputs and outputs in $T_f$ are
single-choices.  This means that, by cycling, $T_f$ performs
indefinitely always the same sequence of input branchings: let $\gamma$ be the 
sequence of the corresponding labels. The inputs accumulated
on the r.h.s. term should be consistent, i.e. must have
labels of the form $l^v_1\cdots l^v_{n_v} = \gamma^i \cdot \iota$, with $\iota$
prefix of $\gamma$.
Consider now the final part~$\iota$ of the sequences
$l^v_1\cdots l^v_{n_v}$ that, as discussed above, are all prefixes of~$\gamma$; as the different prefixes of $\gamma$ are finite, we have
that there exists an infinite subsequence of pairs
$(T_f,\Tbranchsingle {l^v_1}{\ldots\Tbranchsingle
  {l^v_{n_v}}{S_f}\ldots})$
all having the same final $\iota$.  Consider now $maxIn$ as the
maximal number of inputs between two subsequent outputs in $S_f$ (or
its unfolding).  We select from this infinite subsequence a pair
$(T_f,\Tbranchsingle {l^k_1}{\ldots\Tbranchsingle
  {l^k_{n_k}}{S_f}\ldots})$
with $n_k > maxIn$.  We now consider the subsequent pair in the
subsequence
$(T_f,\Tbranchsingle {l^w_1}{\ldots\Tbranchsingle
  {l^w_{n_w}}{S_f}\ldots})$
and the corresponding sequence of rule applications:\\
$\Sigma_k \vdash T_f \subtypet \Tbranchsingle
{l^k_1}{\ldots\Tbranchsingle {l^k_{n_k}}{S_f}\ldots} \rderiv^*$ \\
$\Sigma_w \vdash T_f \subtypet \Tbranchsingle
{l^w_1}{\ldots\Tbranchsingle {l^w_{n_w}}{S_f}\ldots} $\\
Let ${l'_1}\ldots{l'_h}$ be the labels of the outputs involved in
applications of the rule $\textsf{Out}$.  Such outputs are present in
$S_f$ (and its unfoldings). Let
$\Tbranchsingle {l''_1}{\ldots\Tbranchsingle {l''_g}{S'}\ldots}$
(assuming $S' \in\{\oplus, \mu\}$) be the term obtained from $S_f$ by
(unfolding the term and) consuming the outputs labeled with
${l'_1}\ldots{l'_h}$. We have that $S' = S_f$.  We conclude
considering two possible cases: $g \leq |\iota|$ and $g > |\iota|$.
\begin{itemize}
\item
If $g \leq |\iota|$,
we have that all the applications of the rule $\textsf{In}$
involve inputs that are already present in the initial sequence
of inputs of the r.h.s. in
$(T_f,\Tbranchsingle {l^k_1}{\ldots\Tbranchsingle {l^k_{n_k}}{S_f}\ldots})$.
This guarantees that it is possible to apply on the judgement
$\Sigma_f \vdash T_f\subtypet \Tbranchsingle {l^w_1}{\ldots\Tbranchsingle {l^w_{n_w}}{S_f}\ldots}$
the rule $\textsf{Asmp}2$.
As $\textsf{Asmp}2$ has priority it is
necessary to apply this rule thus terminating successfully the
sequence of rule applications. This contradicts the initial assumption
about the infinite sequence of rule applications.
\item
If $g > |\iota|$,
we have that ${l''_1}\ldots {l''_g} = \iota' \cdot \gamma^y \cdot \iota$
with $\iota \cdot \iota' = \gamma$.
From the infinite subsequence we select a pair 
$(T_f,\Tbranchsingle {l^r_1}{\ldots\Tbranchsingle {l^r_{n_r}}{S_f}\ldots})$
such that $n_r-|\iota|$ is greater than the number of applications of the rule 
$\textsf{In}$ in the sequence of rule applications: \\
$\Sigma_k \vdash T_f \subtypet \Tbranchsingle {l^k_1}{\ldots\Tbranchsingle {l^k_{n_k}}{S_f}\ldots} \rderiv^*$ \\
$\Sigma_w \vdash T_f \subtypet \Tbranchsingle {l^w_1}{\ldots\Tbranchsingle {l^w_{n_w}}{S_f}\ldots}$ \\
Consider now the same sequence of rule applications starting from the judgement 
$\Sigma_r \vdash T_f \subtypet \Tbranchsingle {l^r_1}{\ldots\Tbranchsingle {l^r_{n_r}}{S_f}\ldots}$
that introduced 
$(T_f,\linebreak $\mbox{$\Tbranchsingle {l^r_1}{\ldots\Tbranchsingle {l^r_{n_r}}{S_f}\ldots})$}
in the environment. Let \\
$\Sigma_q \vdash T_f \subtypet \Tbranchsingle {l^q_1}{\ldots\Tbranchsingle {l^q_{n_q}}{S''}\ldots}$\\
be the reached judgement. We have that on this judgement it is possible to 
apply $\textsf{Asmp}2$ because $S'' = S_f$ and 
${l^q_1}\ldots{l^q_{n_q}}=\gamma^u\cdot\iota\cdot\iota'\cdot\gamma^y \cdot\iota$.
As discussed in the previous case this contradicts the initial 
assumption on the infinite sequence of rule applications.
\end{itemize}

It remains the final subcase $T_f \in T^{\textsf{out}}$.  In this case
we have $S_f \in T^{\textsf{in}} \cap T^{\textsf{out}}$.  As $T_f$ has
outputs with single-choice, the rule $\textsf{In}$ will be applied at
least once in the sequence of rule applications between every pair and
the subsequent one.  For this reason, it is not restrictive to assume
that in all the pairs
$(T_f,$\linebreak \mbox{$\Tbranchsingle {l^v_1}{\ldots\Tbranchsingle
  {l^v_{n_v}}{S_f}\ldots})$}
the inputs labeled with $l^v_1\cdots l^v_{n_v}$ are produced by
previous unfoldings of the same term $S_f$.  As in $S_f$ all the
inputs and outputs are single-choice, there is a predefined sequence
of inputs within a cycle from $S_f$ to $S_f$ again. Let $\eta$ be
the sequence of the labels corresponding to such inputs. 
Then we have that the inputs accumulated on the l.h.s. term, 
labeled with the sequence of labels $l^v_1\cdots l^v_{n_v}$, will be
such that $l^v_1\cdots l^v_{n_v} = \rho \eta^i$, for some $\rho$ suffix of
$\eta$. We consider a suffix $\rho$ because during the execution of the 
algorithm some previously accumulated input will be consumed to mimick
inputs of the l.h.s. term, and $\rho$ represents the part of inputs
that have been left from a sequence~$\rho$ that has been only partially consumed.
As the suffixes of $\eta$ are finite, we can extract an
infinite subsequence of pairs that always consider the same suffix
$\rho$. Namely, for all the pairs, the second term has an initial
sequence of inputs with a sequence of labels belonging to
$\rho \eta^*$, for the same $\rho$.  But as $\rho$ is a suffix of
$\eta$ there exists a rotation $\gamma'$ of $\eta$ such that all the
sequences belong to $\gamma'^*\rho'$ with $\rho\cdot \rho'=\eta$.  
Hence all the sequences $l^v_1\cdots l^v_{n_v}$ are of the form $\gamma'^i \cdot \rho'$.
By considering $\gamma=\gamma'$ and $\iota=\rho'$ we
can now conclude with the same
arguments used in the previous subcase.  
\qed

We now move to the proof of soundness of the algorithm.

\begin{lemma}\label{lem:soundness}
Given two types 
$T \in T^{\textsf{out}} \cap T^{\textsf{in}}$ (resp. $T \in T^{\textsf{out}}$)
and 
$S \in T^{\textsf{in}}$ (resp. $S \in T^{\textsf{in}} \cap T^{\textsf{out}}$), 
 we have that
  $\emptyset \vdash T \subtypea S \rderiv^* \Sigma' \vdash T'
  \subtypea S' \notrderiv$
  if and only if
  $\emptyset \vdash T \subtypet \decore{S} \rderiv^* \Sigma'' \vdash
  T'' \subtypet S'' \notrderiv$.
\end{lemma} 
\proof We consider the two implications separately starting from the
\emph{if} part.  Assume that
$\emptyset \vdash T \subtypet \decore{S} \rderiv^* \Sigma'' \vdash T''
\subtypet S'' \notrderiv$.
In this sequence of rule applications, the new rules $\textsf{Asmp}2$
and $\textsf{Asmp}3$ are never used otherwise the sequence terminates
successfully by applying such rules. Hence, by applying the same
sequence of rules, we have
$\emptyset \vdash T \subtypea S \rderiv^* \Sigma' \vdash T' \subtypea
S'$
with $T''=T'$, $\undecore{S''}=S'$ and $\undecore{\Sigma''}=\Sigma'$.
We have that $\Sigma' \vdash T' \subtypea S' \notrderiv$, otherwise if
a rule could be applied to this judgement, the same rule could be
applied also to $\Sigma'' \vdash T'' \subtypet S''$ thus contradicting
the assumption \linebreak $\Sigma'' \vdash T'' \subtypet S'' \notrderiv$.

We now move to the \emph{only if} part. Assume the existence of the
sequence of rule applications $\rho_a= 
\emptyset \vdash T \subtypea S \rderiv^* \Sigma' \vdash T' \subtypea S' \notrderiv$.
As the algorithm for $\emptyset \vdash T \subtypet \decore{S}$ 
considers a superset of rules, we have two possible cases: either 
$\emptyset \vdash T \subtypet \decore{S} \rderiv^* \Sigma'' \vdash T'' \subtypet S'' \notrderiv$
by applying the same sequence of rules, or during the 
application of this sequence of rules, starting from $\emptyset \vdash T \subtypet \decore{S}$,
a judgement is reached on which one of the additional rules $\textsf{Asmp}2$
or $\textsf{Asmp}3$ can be applied. 
Namely, there exists a sequence
$\rho_t=\emptyset \vdash T \subtypet \decore{S} \rderiv^* \Sigma_e \vdash T_e \subtypet S_e$,
corresponding to a prefix of $\rho_a$,
such that either $\textsf{Asmp}2$ or $\textsf{Asmp}3$ can be applied
on the judgement $\Sigma_e \vdash T_e \subtypet S_e$.
We conclude the proof by showing that this second case never occurs.
We discuss only $\textsf{Asmp}2$, as the case for
$\textsf{Asmp}3$ is treated similarly.

If $\emptyset \vdash T \subtypet S \rderiv^* \Sigma_e \vdash T_e \subtypet S_e$
and $\textsf{Asmp}2$ can be applied to \mbox{$\Sigma_e \vdash T_e \subtypet S_e$},
there exists an intermediary judgement
$\Sigma_s \vdash T_s \subtypet S_s$, traversed during such sequence
of rule applications, that introduces in the environment the
pair $(T_s,S_s)$ used in the above application of the rule $\textsf{Asmp}2$.
Hence $\rho_t$ has a suffix 
\linebreak $\Sigma_s \vdash T_s \subtypet S_s \rderiv^* \Sigma_e \vdash T_e \subtypet S_e$
with:
\begin{itemize} 
\item $T_s=T_e$,
\item $S_s=\Tbranchsingle {l_1}{\ldots\Tbranchsingle {l_n}{R}\ldots}$ and
$S_e=\Tbranchsingle {l_1}{\ldots\Tbranchsingle {l_m}{R'}\ldots}$
with $\undecore{R}=\undecore{R'}$,
$l_1\cdots l_n = \gamma^i\cdot(l_1\cdots l_{s})$,
$l_1\cdots l_m = \gamma^j\cdot(l_1\cdots l_{s})$
for $i < j$ and $s < |\gamma|$;
\item
during the entire sequence 
$\Sigma_s \vdash T_s \subtypea S_s \rderiv^* 
\Sigma_e \vdash T_e \subtypea S_e$
only a prefix $l_1\cdots l_{r-1}$ of the input actions $l_1\cdots l_n$
is consumed from $S_s$.
\end{itemize} 

Let $\Sigma_s'=\undecore{\Sigma_s}$,
$S_s'=\undecore{S_s}$,
$\Sigma_e'=\undecore{\Sigma_e}$,
and
$S_e'=\undecore{S_e}$.
As $\rho_t$ corresponds to a prefix of $\rho_a$
we have that 
$\rho_a=$\\$\emptyset \vdash T \subtypea S \rderiv^*
\Sigma_s' \vdash T_s \subtypea S_s' \rderiv^* 
\Sigma_e' \vdash T_e \subtypea S_e' \rderiv^*
\Sigma' \vdash T' \subtypea S' \notrderiv$ \\
This is not possible because we now show that after
the sequence $\Sigma_s' \vdash T_s \subtypea S_s'$ \mbox{$ \rderiv^* 
\Sigma_e' \vdash T_e \subtypea S_e'$},
it is not possible to reach any judgement 
$\Sigma' \vdash T' \subtypea S'$ such that $\Sigma' \vdash T' \subtypea S' \notrderiv$.
On the basis of the observations listed above we have:
\begin{itemize} 
\item $T_s=T_e$,
\item $S_s'=\Tbranchsingle {l_1}{\ldots\Tbranchsingle {l_n}{R}\ldots}$,
$S_e'=\Tbranchsingle {l_1}{\ldots\Tbranchsingle {l_m}{R}\ldots}$
with
$l_1\cdots l_n = \gamma^i\cdot(l_1\cdots l_{s})$,
$l_1\cdots l_m = \gamma^j\cdot(l_1\cdots l_{s})$
for $i < j$ and $s < |\gamma|$;
\item
during the entire sequence 
$\Sigma_s' \vdash T_s \subtypea S_s' \rderiv^* 
\Sigma_e' \vdash T_e \subtypea S_e'$
only a prefix $l_1\cdots l_{r-1}$ of the input actions $l_1\cdots l_n$
is consumed from $S_s'$.
\end{itemize}
We have that the sequence of
rules applied in 
$\Sigma_s' \vdash T_s \subtypea S_s' \rderiv^* 
\Sigma_e' \vdash T_e \subtypea S_e'$ 
is the unique one that can be applied also to the ending judgement 
$\Sigma_e' \vdash T_e \subtypea S_e'$.
This is guaranteed by the fact that we are considering
single-choice subtyping, the correspondence of the terms
$T_s=T_e$, the availability of the input actions
labeled with $l_1\cdots l_{r-1}$ in $S_e'$,
and the fact that both $S_s'$ and $S_e'$
have the same term $R$ at the
end of their initial input actions.
Consider now the judgement $\Sigma_e'' \vdash T_e' \subtypea S_e''$
reached at the end of such sequence, i.e.
$\Sigma_e' \vdash T_e \subtypea S_e' \rderiv^* 
\Sigma_e'' \vdash T_e' \subtypea S_e''$. 
We have that the same properties listed above hold also for the 
new sequence $\Sigma_e' \vdash T_e \subtypea S_e' \rderiv^* 
\Sigma_e'' \vdash T_e' \subtypea S_e''$:
\begin{itemize} 
\item $T_e=T_e'$, because the same transformations are
applied to the l.h.s. terms by the rules that are applied. 
\item $S_e'=\Tbranchsingle {l_1}{\ldots\Tbranchsingle {l_m}{R}\ldots}$ and
$S_e''=\Tbranchsingle {l_v}{\ldots\Tbranchsingle {l_m}{R}\ldots}$
with
$l_1\cdots l_m = \gamma^j\cdot(l_1\cdots l_{s})$ and
$l_1\cdots l_v = \gamma^{(j+(j-i))}\cdot(l_1\cdots l_{s})$,
where $i$ is the number of the repetitions of $\gamma$ 
in the initial r.h.s. term $S_s'$, hence $j-i$ is the number
of new repetitions of $\gamma$ added during the sequence of
rule applications.
From the previous properties 
we have $j > i$, hence $j < j+(j-i)$,
and $s < |\gamma|$.
\item
During the entire sequence 
$\Sigma_e' \vdash T_e \subtypea S_e' \rderiv^* 
\Sigma_e'' \vdash T_e' \subtypea S_e''$
only a prefix $l_1\cdots l_{r-1}$ of the input actions $l_1\cdots l_m$
is consumed from $S_e'$.
\end{itemize}
As these properties continue to hold, we have that the 
sequence of rules applied in
$\Sigma_e' \vdash T_e \subtypea S_e' \rderiv^* 
\Sigma_e'' \vdash T_e' \subtypea S_e''$
can be continued to be applied indefinitely,
hence it is not possible to reach any judgement $\Sigma' \vdash T' \subtypea S'$
such that
\mbox{$\Sigma' \vdash T' \subtypea S' \notrderiv$}.
\qed

We can finally conclude with the following theorem that states decidability 
for more general versions of $\selsubtype_\mathsf{sin}$ and $\selsubtype_\mathsf{sout}$, where 
we do not impose related types to belong to $\noinf$.
  
\begin{theorem}[Algorithm Correctness]\label{thm:algcor}
  Given two types 
$T \in T^{\textsf{out}} \cap T^{\textsf{in}}$ (resp.\ $T \in T^{\textsf{out}}$)
and 
$S \in T^{\textsf{in}}$ (resp.\ $S \in T^{\textsf{in}} \, \cap\,  T^{\textsf{out}}$), 
 we have that \mbox{$\emptyset \vdash T \subtypet \decore{S}$}
  if and only if
  $T \leq S$. 
\end{theorem} 
\proof It follows immediately from Lemma~\ref {lem:semidecidable},
Lemma~\ref{lemma:termination} and Lemma~\ref{lem:soundness}. \qed


As an obvious consequence of algorithm correctness, we have that the two relations 
$\selsubtype_{\mathsf{sout}}$ and $\selsubtype_{\mathsf{sin}}$ are decidable.

\begin{corollary}
The asynchronous single-choice output relation $\selsubtype_{\mathsf{sout}}$
and the asynchronous single-choice input relation $\selsubtype_{\mathsf{sin}}$
are 
decidable.
\end{corollary}
\proof
In order to verify whether $T \, \selsubtype_{\mathsf{sin}} \, S$ it is sufficient
to check that $T \in T^{\textsf{in}} \cap T^{\textsf{out}} \cap \noinf$
and $S \in T^{\textsf{in}} \cap \noinf$ and then verify whether $T \, \subtype \, S$,
which is decidable for terms belonging to these sets as proved in Theorem~\ref{thm:algcor}.
For $T \, \selsubtype_{\mathsf{sout}} \, S$ it is possible to proceed in the same
way with the only difference that
the check is whether $T \in T^{\textsf{out}} \cap \noinf$
and $S \in T^{\textsf{in}} \cap T^{\textsf{out}} \cap \noinf$.
\qed

Notice that with respect to the general case,
when our algorithm is applied to the restricted case of 
$\selsubtype_{\mathsf{sout}}$ and $\selsubtype_{\mathsf{sin}}$,
rule $\textsf{Asmp}3$ and premise $\& \in T$ of 
rule $\textsf{Asmp}2$ become useless. This follows from the fact
that types in $\noinf$ never satisfy the premises of $\textsf{Asmp}3$
and always satisfy the premise $\& \! \in \! T$ of 
rule $\textsf{Asmp}2$.

%
%
%
%
%

%

%% file: conclusions.tex
\paragraph{Related Work}
\label{sec:concl}
Lange and Yoshida~\cite{LY16} have independently and simultaneously
provided an undecidability result for a class of communicating
automata called asynchronous duplex systems (which are shown to
correspond to a class of binary session types). They prove that
automata compatibility (checking whether two automata in parallel can
safely interact) is undecidable and then show that such result makes
also asynchronous session types subtyping undecidable. 
Their proof consists of an encoding of the termination problem for Turing machines (rather than our simpler and direct encoding based on queue machines) into deciding automata compatibility.
Most importantly, in order to prove undecidability of
a subtyping $\preceq$, 
a translation $\mathcal M$ from types to automata is used which
exploits the following result (Theorem 5.1, \cite{LY16}):
\begin{quote}
  Given two types $T$ and $S$, we have that $T\preceq S$ if and only
  if $\mathcal M(T)$ is compatible with $\mathcal M(\overline S)$. 
\end{quote}
Above, $\overline S$ is the dual of $S$, the type obtained from $S$ by
inverting inputs (outputs) with outputs (inputs).  As compatibility
between communicating automata is a symmetric relation, we have that
this approach to prove undecidability of subtyping can be applied only
to dual closed relations $\preceq$, i.e., $T \preceq S$ if and only if
$\overline S \preceq \overline T$.  Among all the subtyping relations
in the literature (discussed in details in Section \ref{subsec:impact}), 
this property holds only for the definition
of subtyping by Chen et al.~\cite{MariangiolaPreciness}, where orphan
messages are not admitted.  For instance, we have that the type
$\Trec t.\Tselectsimple{l}{\Tvar t}$ is a subtype of
$\Trec t.\Tbranchsimple{l'}{\Tselectsimple{l}{\Tvar t}}$ but obviously
$\Trec t.\Tselectsimple{l'}{\Tbranchsimple{l}{\Tvar t}}$ is not a
subtype of $\Trec t.\Tbranchsimple{l}{\Tvar t}$.
Moreover, differently from Lange and Yoshida~\cite{LY16} we show that, to get undecidability, it is sufficient to
consider a restricted version of asynchronous binary session subtyping ($\selsubtype$ relation) that is much less expressive and
that cannot be further simplified by imposing limitations on the branching/selection structure of types 
(otherwise it becomes decidable).
As shown, this has one plain advantage: it allowed us to easily show undecidability of various existing subtyping relations.

Concerning decidability,
Lange and Yoshida~\cite{LY16} independently and simultaneously
proved a result similar to ours (Theorem \ref{thm:algcor}).
They present an algorithm for deciding subtyping between types $T$ and $S$, with
one between $T$ and $S$ being a single-choice session type (i.e.\ a type where every branching/selection has a single choice) and the other one being a general (any) session
type. In particular, they consider a dual closed subtyping relation 
defined following the orphan-message free approach by Chen et al.~\cite{MariangiolaPreciness}.
On the other hand, we show 
(Theorem \ref{thm:algcor})
$T \,\subtype\, S$ to be decidable when: one between $T$ and $S$ is a single-choice session type (as for Lange and Yoshida) and the other one is either
a type with single outputs, if it is $T$, or a type with single inputs, if it is $S$.
We first notice that our $\subtype$ relation 
is more general 
with respect to the one used by Lange and Yoshida 
in that it does not include the orphan message free constraint. 
Moreover, although it may seem that Lange and Yoshida 
can effectively relate 
types that do not fall under our (more restricted) syntactical characterization,
covariance and contravariance prevent any type containing at
least one non-single input branch and 
 one non-single output selection (both
reachable in the subtyping simulation game) to be related with a
single-choice type.
We finally observe that 
since, 
differently from the relation used by Lange and Yoshida, 
our $\subtype$ relation is not dual closed,
we need to explicitly carry out two separate proofs for the two cases of $T \,\subtype\, S$:
one where the single-choice session type is $T$ and another one where it is, instead,~$S$.





%
%

\paragraph{Conclusion}
We have proven that asynchronous subtyping for session types is
undecidable. Moreover, we have shown that subtyping becomes decidable
if we put some restrictions on the branching/selection structure. 
As future work, we plan to search for alternative subtyping relations
that enjoy properties similar to $\subtype$, but are decidable.

%